\documentclass[11pt]{article}
\usepackage[margin=1in]{geometry}
\usepackage{amsmath,amsfonts,amssymb,amsthm}
\usepackage{graphicx}
\usepackage{enumerate}
\usepackage{bbm}
\usepackage{verbatim}
\usepackage{hyperref,color}
\usepackage[capitalize,nameinlink]{cleveref}
\usepackage[dvipsnames]{xcolor}
\hypersetup{
	colorlinks=true,
	pdfpagemode=UseNone,
	citecolor=OliveGreen,
	linkcolor=NavyBlue,
	urlcolor=Magenta,
	pdfstartview=FitW
}
\usepackage{appendix}
\usepackage{makecell}
\usepackage{tablefootnote}
\crefname{appsec}{Appendix}{Appendices}
\usepackage{tikz}
\usepackage{pgfplots}
\usepackage{xifthen}
\usepackage{subcaption}
\usepackage{hhline}

\theoremstyle{plain}
\newtheorem{theorem}{Theorem}[section]
\newtheorem{proposition}[theorem]{Proposition}
\newtheorem{lemma}[theorem]{Lemma}
\newtheorem{corollary}[theorem]{Corollary}

\theoremstyle{definition}
\newtheorem{definition}[theorem]{Definition}
\newtheorem{example}[theorem]{Example}

\newtheorem*{assumption*}{Assumption}

\theoremstyle{remark}
\newtheorem{remark}[theorem]{Remark}

\crefname{lemma}{Lemma}{Lemmas}
\crefname{theorem}{Theorem}{Theorems}
\crefname{definition}{Definition}{Definitions}
\crefname{fact}{Fact}{Facts}
\crefname{claim}{Claim}{Claims}
\crefname{proposition}{Proposition}{Propositions}

\newcommand{\dif}{\,\mathrm{d}}
\newcommand{\Cov}{\mathrm{Cov}}

\newcommand{\ind}{\*1}

\newcommand{\norm}[1]{\left\lVert #1 \right\rVert}

\newcommand{\diag}{\mathrm{diag}}

\newcommand{\eps}{\varepsilon}

\newcommand{\N}{\mathbb{N}}

\newcommand{\R}{\mathbb{R}}

\newcommand{\II}{\mathcal{I}}

\newcommand{\MM}{\mathcal{M}}

\newcommand{\XX}{\mathcal{X}}

\newcommand{\TV}[2]{d_{\mathrm{TV}}\left({#1},\,{#2}\right)}
\newcommand{\e}{\mathrm{e}}
\renewcommand{\epsilon}{\varepsilon}
\renewcommand{\emptyset}{\varnothing}
\newcommand{\set}[1]{\left\{#1\right\}}
\newcommand{\tuple}[1]{\left(#1\right)} 
\newcommand{\inner}[2]{\left\langle #1,#2\right\rangle}
\newcommand{\tp}{\tuple}
\newcommand{\ol}{\overline}
\newcommand{\abs}[1]{\left\vert#1\right\vert}

\def\*#1{\boldsymbol{#1}} 
\def\+#1{\mathcal{#1}} 
\def\-#1{\mathrm{#1}} 
\def\=#1{\mathbb{#1}} 
\def\!#1{\mathfrak{#1}} 

\def\oPr{\mathop{\mathrm{Pr}}}
\renewcommand{\Pr}[2][]{ \ifthenelse{\isempty{#1}}
  {\oPr\left[#2\right]}
  {\oPr_{#1}\left[#2\right]} } 

\def\oE{\mathop{\mathbb{E}}}
\newcommand{\E}[2][]{ \ifthenelse{\isempty{#1}}
  {\oE\left[#2\right]}
  {\oE_{#1}\left[#2\right]} }

\def\oVar{\mathrm{Var}}
\newcommand{\Var}[2][]{ \ifthenelse{\isempty{#1}}
  {\oVar\left[#2\right]}
  {\oVar_{#1}\left[#2\right]} }

\def\oEnt{\mathrm{Ent}}
\newcommand{\Ent}[2][]{ \ifthenelse{\isempty{#1}}
  {\oEnt\left[#2\right]}
  {\oEnt_{#1}\left[#2\right]} }

\newcommand{\PhiEnt}[2][]{ \ifthenelse{\isempty{#1}}
  {\oEnt^\phi\left[#2\right]}
  {\oEnt^\phi_{#1}\left[#2\right]} }

\newcommand{\mathsc}[1]{{\normalfont\textsc{#1}}}

\title{Rapid Mixing on Random Regular Graphs beyond Uniqueness}

\author{Xiaoyu Chen\thanks{Massachusetts Institute of Technology, Cambridge, Massachusetts, USA. Email: \texttt{xiaoyu@mit.edu}.} 
\and Zejia Chen\thanks{Georgia Institute of Technology, Atlanta, Georgia, USA. Email: \texttt{zchen3091@gatech.edu}.}
\and Zongchen Chen\thanks{Georgia Institute of Technology, Atlanta, Georgia, USA. Email: \texttt{chenzongchen@gatech.edu}.}
\and Yitong Yin\thanks{Nanjing University, Nanjing, Jiangsu, China. Email: \texttt{yinyt@nju.edu.cn}.}
\and Xinyuan Zhang\thanks{Nanjing University, Nanjing, Jiangsu, China. Email: \texttt{zhangxy@smail.nju.edu.cn}.}
}
\date{\today}

\pgfplotsset{compat=1.18} 
\begin{document}

\maketitle

\begin{abstract}

The hardcore model is a fundamental probabilistic model extensively studied in statistical physics, probability theory, and computer science.
It defines a Gibbs distribution over independent sets of a given graph, parameterized by a vertex activity $\lambda > 0$. 
For graphs of maximum degree $\Delta$, a well-known computational phase transition occurs at the tree-uniqueness threshold $\lambda_c(\Delta) = \frac{(\Delta-1)^{\Delta-1}}{(\Delta-2)^\Delta}$, where the mixing behavior of the Glauber dynamics (a simple Markov chain) undergoes a sharp transition: it mixes in nearly linear time for $\lambda < \lambda_c(\Delta)$, in polynomial but super-linear time at $\lambda = \lambda_c(\Delta)$, and experiences exponential slowdown for $\lambda > \lambda_c(\Delta)$.

It is conjectured that random regular graphs exhibit different mixing behavior,
with the slowdown occurring far beyond the uniqueness threshold.
We confirm this conjecture by showing that,
for the hardcore model on random $\Delta$-regular graphs, 
the Glauber dynamics mixes rapidly with high probability when $\lambda = O(1/\sqrt{\Delta})$, which is significantly beyond the uniqueness threshold  $\lambda_c(\Delta) \approx e/\Delta$.
Our result establishes a sharp distinction between the hardcore model on worst-case and beyond-worst-case instances, showing that the worst-case and average-case complexities of sampling and counting are fundamentally different.

This result of rapid mixing on random instances follows from a new criterion we establish for rapid mixing of Glauber dynamics for any distribution supported on a downward closed set family. 
Our criterion is simple, general, and easy to check.
In addition to proving new mixing conditions for the hardcore model, we also establish improved mixing time bounds for sampling uniform matchings or $b$-matchings on graphs, the random cluster model on matroids with $q \in [0,1)$, and the determinantal point process.
Our proof of this new criterion for rapid mixing combines and generalizes several recent tools in a novel way, including a trickle-down theorem for field dynamics, spectral/entropic stability, and a new comparison result between field dynamics and Glauber dynamics.

\bigskip

\end{abstract}

\thispagestyle{empty}

\newpage 
\thispagestyle{empty}

\tableofcontents

\newpage

\setcounter{page}{1}

\newpage
\section{Introduction}
The hardcore model, first introduced in statistical physics to describe systems of non-overlapping gas particles, has become a fundamental example of hard-constraint undirected graphical models.
Given a graph $G=(V,E)$, let $\II = \II(G)$ denote the set of all independent sets of $G$, where an independent set is a subset of vertices that induces no edges.
The model is parameterized by $\lambda > 0$, known as the \emph{fugacity}, 
and defines the Gibbs distribution $\mu = \mu_{G,\lambda}$ over $\II$ as
\begin{align*}
    \mu(S) := \frac{\lambda^{|S|}}{Z}, \quad \forall S \in \II,
\end{align*}
where $Z = Z_{G,\lambda} := \sum_{S \in \II} \lambda^{|S|}$ is the normalizing constant, known as the \emph{partition function}.


A fundamental computational problem is sampling from the Gibbs distribution, which is poly-time equivalent to estimating the partition function $Z$ via standard reductions.

Over the past two decades, a computational phase transition for approximate counting and sampling has been rigorously established through a series of influential works. 
This phase transition occurs at a critical threshold, $\lambda_c(\Delta) = \frac{(\Delta-1)^{\Delta-1}}{(\Delta-2)^\Delta}$, 
known as the \emph{tree-uniqueness threshold}, which determines the uniqueness of the Gibbs measure on the infinite $\Delta$-regular tree. 
Specifically, when $\lambda \le \lambda_c(\Delta)$, the Gibbs measure is unique, whereas for $\lambda > \lambda_c(\Delta)$, multiple Gibbs measures coexist.
Remarkably, this threshold also dictates the computational tractability of the hardcore model on graphs with maximum degree $\Delta$.
When $\lambda < \lambda_c(\Delta)$, deterministic FPTASes exist for the partition function~\cite{weitz2006counting,peters2019conjecture}. 
In contrast, for $\lambda > \lambda_c(\Delta)$,  approximate counting and sampling become intractable in polynomial time unless  $\textsf{NP}=\textsf{RP}$ \cite{sly2010computational,sly2014counting,galanis2016inapproximability}.
%

A sharper and more practical computational phase transition concerns the \emph{Glauber dynamics} (also known as the \emph{Gibbs sampler}), a canonical Markov Chain Monte Carlo (MCMC) algorithm for sampling from high-dimensional distributions. 
At each step, the Glauber dynamics picks a vertex $v \in V$ uniformly at random and updates the current independent set $S_t \in \II$ as follows: 
Let $S' := S_t \setminus \{v\}$. If $S' \cup \{v\} \notin \II$, set $S_{t+1} := S'$; otherwise, set $S_{t+1} := S' \cup \{v\}$ with probability $\frac{\lambda}{1+\lambda}$ and $S_{t+1} := S'$ with the remaining probability.
This well-known Markov chain, denoted by $P_\mathsc{gd}$, is reversible and converges to the hardcore Gibbs distribution.
Its mixing time is defined as
\begin{align*}
    T_{\mathrm{mix}}(P_\mathsc{gd}) := \max_{S \in \II} \min_{t \in \N} \left\{ \TV{P_\mathsc{gd}^t (S,\cdot)}{\mu} \le 1/4 \right\}.
\end{align*}
For the hardcore model, 
the mixing behavior of Glauber dynamics undergoes a sharp transition at the tree-uniqueness threshold.
Starting with \cite{ALO2021spectral}, a recent line of works \cite{CLV2023rapid, CLV2023optimal, blanca2022mixing, CFYZ2024rapid, AJKPV2022entropic, CFYZ22optimal, chen2022localization} have established that for $\lambda < \lambda_c(\Delta)$, the mixing time is $O(n \log n)$, where $n$ is the number of vertices. 
In contrast, for $\lambda > \lambda_c(\Delta)$, the mixing time can grow as large as $\exp(\Omega(n))$ in the worst case \cite{mossel2009hardness}.
Most recently, the mixing time at criticality was settled in \cite{chen2024rapid}, showing that for $\lambda = \lambda_c(\Delta)$, the mixing time is at most $O(n^{2+4e + O(1/\Delta)})$ and at least $\Omega(n^{4/3})$ in the worst case.

A key property at the heart of this computational phase transition is the decay of correlation.
In the uniqueness regime, the hardcore model exhibits \emph{strong spatial mixing} \cite{weitz2006counting} and \emph{$\ell_\infty$-spectral independence} \cite{CLV2023rapid,CFYZ2024rapid} (also known as \emph{total influence decay}). 
Both properties characterize correlation decay, asserting that vertices far apart tend to have diminishing influence on each other under arbitrary pinning.
These correlation decay properties form the foundation of previous deterministic approximate counting algorithms and ultimately lead to the rapid mixing of Glauber dynamics via the \emph{spectral independence} framework~\cite{ALO2021spectral,alev2020improved}.
In contrast, both strong spatial mixing and total influence decay cease to exist in the supercritical regime, i.e., when $\lambda > \lambda_c(\Delta)$,
which is evident on regular trees by examining the influence between the leaves and the root.

These computational phase transition results hold for arbitrary (or worst-case) graphs with maximum degree $\Delta$.
For the hardcore model, and more generally anti-ferromagnetic 2-spin models, 
the hard instances are given by random $\Delta$-regular \emph{bipartite} graphs,
whose Gibbs measures converge locally to \emph{semi}-translation-invariant Gibbs measures 
on the infinite $\Delta$-regular tree \cite{sly2014counting}.
%
%

We study computational phase transitions from a \emph{beyond-worst-case} perspective 
and investigate how far rapid mixing can extend beyond the worst-case critical threshold, 
specifically for the average-case instances that evade worst-case hardness. 
For the hardcore model, a canonical class of average-case instances is random $\Delta$-regular graphs. 
Statistical physicists predict that the mixing behavior on these graphs differs significantly from worst-case instances, 
as weaker notions of correlation decay persist beyond the uniqueness threshold, which are presumed to better characterize the rapid mixing of Glauber dynamics.
Prominent examples include replica symmetry and non-reconstructability. 
The recently-introduced and powerful notion of spectral independence is also conjectured to hold beyond uniqueness on random regular graphs as well. 

One explanation of these predictions is that, due to the locally tree-like structure, the mixing behavior on random regular graphs should resemble that of complete regular trees.
It is known that in the replica symmetric regime, the local measure around a random vertex in a random regular graph converges weakly to the translation-invariant Gibbs measure on the infinite regular tree  \cite{bhatnagar2016decay}.
Furthermore, evidence suggests that the mixing time of the Glauber dynamics on complete regular trees is intimately tied to a critical threshold $\lambda_r(\Delta)$, marking the transition from non-reconstruction to reconstruction \cite{martinelli2004glauber, martinelli2007fast, sly2017glauber}.
This leads to a natural, albeit optimistic, conjecture: Glauber dynamics should mix rapidly on random regular graphs up to the reconstruction threshold $\lambda_r(\Delta)$, which is significantly higher than the uniqueness threshold (see \cite{bhatnagar2016decay}).
%
However, to the best of our knowledge, no algorithmic result or proof of spectral independence exists for the hardcore model on random regular graphs, even slightly above the uniqueness threshold $\lambda_c(\Delta) \approx e/\Delta$.

In this paper, we take the first step and show that the Glauber dynamics mixes rapidly on random $\Delta$-regular graphs when $\lambda = O(1/\sqrt{\Delta})$, which is significantly beyond the uniqueness threshold. 

\begin{theorem}\label{thm:HC-main}
     Let $\Delta \ge 3$ be an integer, and let $\delta \in (0,1)$. 
     Consider the Glauber dynamics for the hardcore model on a random $\Delta$-regular graph $G$ on $n$ vertices with fugacity $\lambda$.
     If $\lambda\le \frac{1-\delta}{2\sqrt{\Delta-1} - 1}$,
     then with high probability over the choice of $G$, the mixing time satisfies:
    \begin{align*}
        T_{\mathrm{mix}}(P_\mathsc{gd}) = O\left( \frac{n^2}{\delta} \log \Delta \right).
    \end{align*}
\end{theorem}


To develop an intuitive understanding of the bound on $\lambda$ in \cref{thm:HC-main}, 
we parameterize the hardcore model using the \emph{occupancy fraction} $\alpha := \frac{1}{n} \mathbb{E}_\mu[|S|]$, which represents the expected fraction of occupied vertices in an independent set sampled from the Gibbs measure. 
On random $\Delta$-regular graphs, $\alpha$  asymptotically equals the probability that the root is occupied in the translation-invariant Gibbs measure on the infinite $\Delta$-regular tree, provided $\lambda$ is not too large; see \cite{bhatnagar2016decay} for an explicit relation between $\lambda$ and $\alpha$ on the infinite tree.

It is well known that the normalized independence number (i.e., the maximum size of an independent set divided by $n$) of a random $\Delta$-regular graph satisfies $\alpha_{\max}(\Delta) \sim \frac{2 \log \Delta}{\Delta}$ with high probability. 

At the reconstruction threshold $\lambda_r(\Delta)$, the occupancy fraction is approximately $\alpha_r(\Delta) \sim \frac{\log \Delta}{\Delta}$ \cite{bhatnagar2016decay}. 
Notably, this density, $\frac{\log \Delta}{\Delta}$, matches the best known efficient algorithm for finding large independent sets in random $\Delta$-regular graphs. Furthermore, strong evidence suggests that solving the problem beyond this threshold is computationally intractable \cite{COE15,GS17,RV17}; see also recent relevant results on Erd\H{o}s--R\'enyi random graphs \cite{GJW24,Wei22,HS25}.

At fugacity $\lambda^*(\Delta) \sim \frac{1}{2\sqrt{\Delta}}$ as in \cref{thm:HC-main}, we can sample independent sets from the hardcore model with occupancy fraction $\alpha^*(\Delta) \sim \frac{\log \Delta}{2 \Delta}$, which is a factor of 2 away from the conjectured reconstruction threshold $\alpha_r(\Delta)$.
As a baseline, at the uniqueness threshold $\lambda_c(\Delta)$, the occupancy fraction satisfies $\alpha_c(\Delta) \sim \frac{1}{\Delta}$, thus much smaller.
In this sense, \cref{thm:HC-main} is analogous to previous sampling results on random instances, such as the antiferromagnetic Ising model on random regular graphs \cite{koehler2022sampling, anari2024trickle}, the Sherrington--Kirkpatrick (SK) model \cite{eldan2022spectral, anari2024trickle}, and the $p$-spin model \cite{adhikari2024spectral,anari2024universality,mikulincer2024stochastic}.
In both cases, rapid mixing was established for inverse temperatures way beyond the corresponding critical temperatures, though still not extending across the entire regime of correlation decay (characterized by non-reconstruction or replica symmetry).

\cref{thm:HC-main} follows from a more general theorem, which relates the rapid mixing of Glauber dynamics for the hardcore model to the {minimum}  eigenvalue of the underlying graph.

\begin{theorem}\label{thm:HC-main-lambda-min}
    Let $\delta \in (0,1)$,
    and consider the Glauber dynamics for the hardcore model on a graph $G$ with $n$ vertices and fugacity $\lambda$.
    If
    \begin{align*}
        \lambda \le \frac{1-\delta}{- \lambda_{\min}(A_G) - 1},
    \end{align*}
    where $\lambda_{\min}(A_G)$ is the minimum eigenvalue of the adjacency matrix $A_G$ of $G$, 
    then the mixing time satisfies $T_{\mathrm{mix}}(P_\mathsc{gd}) = O\left(\frac{n^2}{\delta} \log \frac{1}{\lambda} \right)$.
\end{theorem}

For any non-empty graph $G$, the {minimum}  eigenvalue $\lambda_{\min}(A_G)$ ranges over $[-1,-\Delta]$, where $\Delta=\Delta_G$ is the maximum degree.
The extremal values are attained when $G$ is a bipartite graph ($\lambda_{\min}(A_G)=-\Delta$) or a clique ($\lambda_{\min}(A_G)=-1$).
More generally, $\lambda_{\min}(A_G)$ quantifies the extent to which $G$ deviates from being bipartite.
Thus, \cref{thm:HC-main-lambda-min} establishes rapid mixing for the hardcore model in terms of the graph’s non-bipartiteness.

We establish \cref{thm:HC-main-lambda-min} by introducing a simple and universal criterion for the rapid mixing of Glauber dynamics, which is formally stated in \cref{subsec:main-thm}. 
When applied to the hardcore model, this criterion yields the condition in \cref{thm:HC-main-lambda-min} in terms of the minimum graph eigenvalue.
Other notable applications of this approach include uniform matching and $b$-matching.

\begin{theorem}\label{thm:matching-main}
    Let $G$ be a simple graph with $n$ vertices, $m$ edges, and maximum degree $\Delta$.
    Let $\mu_{\mathcal{M}(G)}$ and $\mu_{\mathcal{M}_b(G)}$ denote the uniform distributions over matchings and $b$-matchings in $G$, respectively. 
    \begin{itemize}
        \item The mixing time of the Glauber dynamics for $\mu_{\mathcal{M}(G)}$ is  $O\left(\sqrt{\Delta}mn\log n\right)$. 
        \item The mixing time of the Glauber dynamics for $\mu_{\mathcal{M}_b(G)}$ is $O\left(\min\{\Delta^b,bn\}\cdot bmn\log n\right)$.
    \end{itemize}
\end{theorem}
This represents the first improvement in the mixing time for matchings on general graphs since the seminal work of Jerrum and Sinclair \cite{jerrum1989approximating,sinclair1992improved}.
More broadly, we prove rapid mixing for other applications, 
including the random cluster model on matroids; 
see \cref{subsec:applications} for discussion.

\begin{remark}
We note that a concurrent work \cite{chen2025faster}, which shares a subset of authors with the current paper, establishes mixing time bounds of $\widetilde{O}_\lambda(\Delta^2 m)$ for the Jerrum-Sinclair chain and $\widetilde{O}_\lambda(\Delta^3 m)$ for Glauber dynamics, both on the general monomer-dimer model with edge weight $\lambda>0$. Their approach extends the canonical path method, which is very different from ours. 
\end{remark}



Our proof approach is based on a \textit{trickle-down theorem for linear-tilt localization schemes} established in \cite{anari2024trickle}. 
By applying it to the field dynamics, we derive new mixing results for the hardcore model and many other applications; see \cref{subsec:proof-outline} for the proof outline. 
We note that the recent work \cite{anari2024trickle} established an analogous trickle-down theorem for stochastic localization and applied it to the Ising and SK models.

\subsection{A new criterion for rapid mixing of Glauber dynamics}
\label{subsec:main-thm}

Our new criterion for rapid mixing is stated for distributions supported on down-closed set families. 

Let $\XX \subseteq \set{0,1}^n$  correspond to a non-empty set family that is \emph{downward closed} (also called \emph{lower set}, \emph{abstract simplicial complex}, \emph{independence system}, etc.), i.e., if $T \in \XX$ and $S \subseteq T$, then $S \in \XX$.
For each $S \in \XX$, define $V_S = \{i \in [n] \setminus S: S\cup \{i\} \in \XX\}$ as the set of elements that can be added to $S$ while still remaining in $\XX$.
A set $S \in \XX$ is called \emph{maximal} if there is no such $T\in\XX$ that $S\subset T$, i.e., $V_S = \emptyset$.

For each $S \in \XX$, let $\XX^S = \{T \subseteq [n] \setminus S: S \cup T \in \XX\}$ be the induced set family on the ground set $[n] \setminus S$.
If $\XX$ is downward closed, then $\XX^S$ is also downward closed.
For example, if $([n], \XX)$ forms a matroid, then $([n] \setminus S, \XX^S)$ corresponds to the matroid contraction.

Suppose $\mu$ is a distribution fully supported on $\XX$, 
i.e., $\mu(S) > 0$ for all $S \in \XX$.
For each $S \in \XX$, 
define $\mu^S$ as the conditional distribution supported on $\XX^S$ s.t.~$\mu^S(T) \propto \mu(S \cup T)$ for each $T \in \XX^S$.

The Glauber dynamics for $\mu$ is then defined as follows.
In each transition step, the chain picks an $i \in [n]$ uniformly at random and updates the current state $S_t$ according to the following rule:
\begin{itemize}
    \item Let $S' := S_t \setminus \{i\}$. If $S' \cup \{i\} \notin \XX$, then set $S_{t+1} \gets S'$.
    \item Otherwise, set  $S_{t+1} \gets S' \cup \{i\}$ with probability $\frac{\mu(S' \cup \{i\})}{\mu(S') + \mu(S' \cup \{i\})}$ and set $S_{t+1} \gets S'$ with remaining probability $\frac{\mu(S')}{\mu(S') + \mu(S' \cup \{i\})}$.
\end{itemize}
Since the state space $\XX$ is downward closed, it is straightforward to verify that the Glauber dynamics is ergodic and reversible with respect to the (unique) stationary distribution $\mu$. 

The statement of our criterion for rapid mixing still requires the following two definitions.

\begin{definition}[Marginal ratios]\label{def:marginal-ratios}
    For each non-maximal $S \in \XX$, define the vector of \emph{marginal ratios} $\*r_S \in \R^{V_S}$ at $S$ by
    \begin{align*}
        \*r_S(i) := \frac{\mu^S(\{i\})}{\mu^S(\emptyset)}
        = \frac{\mu(S \cup \{i\})}{\mu(S)}, 
        \quad \forall i \in V_S.
    \end{align*}
    Furthermore, define the \emph{maximum marginal ratio} as
        $r_{\max} = \max_{S \in \XX} \; \max_{i \in V_S} \; \*r_S(i)$.
\end{definition}
The marginal ratio of $i \in V_S$ at $S \in \XX$ is the ratio of probability masses when adding $i$ to $S$. Under the binary indicator representation, it can be interpreted as the marginal ratio of the $i$-th coordinate conditioned on all coordinates in $S$ being $1$ and all coordinate in $[n] \setminus S \setminus \{i\}$ being $0$.

\begin{definition}[Pairwise dependency matrix]
    For each non-maximal $S \in \XX$, define the \emph{pairwise dependency matrix} $M_S \in \R^{V_S \times V_S}$ at $S$ by
    \begin{align*}
        M_S(i,j) := 
        \begin{cases}
            \frac{\mu^S(\{i,j\}) \mu^S(\emptyset)}{\mu^S(\{i\}) \mu^S(\{j\})} - 1, \quad& i \neq j;\\
            0, \quad & i = j.
        \end{cases}
    \end{align*}
\end{definition}

This notion of pairwise dependency matrix captures the strength of dependency between variables in a local sense.
Specifically, these matrices have the following nice properties. 

\begin{remark}\label{remark:invariant}
    The pairwise dependency matrix $M_S$ is symmetric. Furthermore, it is invariant under any external fields (also called exponential tilt).
    For any $\*\lambda \in \R_{>0}^n$, define the distribution $\*\lambda * \mu$ as 
    \begin{align*}
        (\*\lambda * \mu) (S) \propto \mu(S) \prod_{i \in S} \lambda_i, \quad \forall S \in \XX. 
    \end{align*}
    If $\*\lambda = \lambda \*1$, we write $\lambda * \mu = \lambda \*1 * \mu$ for simplicity.
    Then, it is straightforward to verify that
    \begin{align*}
        M_S(\*\lambda * \mu) = M_S(\mu).
    \end{align*}
\end{remark}

In \cref{subsec:applications}, we will further show by examples that, in many applications, the pairwise dependency matrices are easy to calculate and well-structured.

We are now ready to state our main theorems for general distributions.

\begin{theorem}\label{thm:main-0}
    Let $\XX \subseteq \{0,1\}^n$ be a non-empty downward closed set family, 
    and suppose $\mu$ is a distribution fully supported on $\XX$.
    If the following condition holds:
    \begin{align}\label{eq:strong-assump-main}
        M_S \preceq I, \quad \text{for each non-maximal $S \in \XX$},
    \end{align}
    then the mixing time of the Glauber dynamics for $\mu$ is upper bounded by
    \begin{align*}
        T_{\mathrm{mix}}(P_\mathsc{gd}) = O\left( (1+r_{\max})n \log\log\left( \frac{1}{\mu_{\min}} \right) \right),
    \end{align*}
    where $\mu_{\min} = \min_{S \in \XX} \mu(S)$.
\end{theorem}

Condition~\eqref{eq:strong-assump-main} is equivalent to that the generating polynomial of $\mu$ is strongly log-concave, see \cref{prop:log-concave}. Hence, we can state the assumption of \cref{thm:main-0} as a strong log-concavity condition.
We refer to \cref{subsec:log-concave} for more discussions.

\begin{theorem}\label{thm:main-delta}
    Let $\XX \subseteq \{0,1\}^n$ be a non-empty downward closed set family, 
    and suppose $\mu$ is a distribution fully supported on $\XX$.
    Let $\delta \in (0,1)$.
    If the following condition holds:
    \begin{align}\label{eq:weak-assump-main}
        M_S \preceq I + (1-\delta) \diag(\*r_S)^{-1}, \quad \text{for each non-maximal $S \in \XX$},
    \end{align}
    then the mixing time of the Glauber dynamics for $\mu$ is upper bounded by
    \begin{align*}
        T_{\mathrm{mix}}(P_\mathsc{gd}) = O\left( (1+r_{\max}) \frac{n}{\delta} \log\left( \frac{1}{\mu_{\min}} \right) \right),
    \end{align*}
    where $\mu_{\min} = \min_{S \in \XX} \mu(S)$.
\end{theorem}

In \cref{thm:main-0}, the mixing time bound is derived using the modified log-Sobolev inequality, resulting in a double-logarithmic dependency on $1/\mu_{\min}$.
In \cref{thm:main-delta}, the mixing time is derived using the Poincar\'{e} inequality (i.e., the spectral gap), leading to a logarithmic dependency on $1/\mu_{\min}$.
We also remark that since both Conditions~\eqref{eq:strong-assump-main} and \eqref{eq:weak-assump-main} considers conditioning of any non-maximal $S \in \XX$, we can obtain from \cref{thm:main-0,thm:main-delta} efficient approximate counting algorithms for the associated partition functions in many applications via self-reducibility.

From the perspective of simplicial complexes, previous results have extensively studied the rapid mixing of down-up walks for distributions supported over maximal faces of a pure simplicial complex \cite{ALOVii,Opp18,alev2020improved}. 
Meanwhile, previous works studied the Glauber dynamics by transferring the target distribution on a product space (e.g., $\{0,1\}^n$) into a distribution over the maximal faces of a pure simplicial complex by considering all coordinate-label pairs; see \cite{ALO2021spectral}.
In contrast, our criterion applies to distributions over \emph{all} faces of the simplicial complex, not limited to the maximal faces, and works for non-pure simplicial complexes as well, where the standard down-up walk is not applicable. 
\subsection{Examples and applications}
\label{subsec:applications}

We provide several examples that serve as applications of \cref{thm:main-0,thm:main-delta}, highlighting the universality of these theorems.

\begin{example}[Product distribution]
    As a toy example, let $\*p =(p_1, \dots, p_n)\in (0,1)^n$, and consider the product distribution with mean vector $\*p$, defined as
     \begin{align*}
        \mu(S) = \prod_{i \in S} p_i \prod_{j \in [n] \setminus S} (1-p_j), \quad \forall S \subseteq [n].
    \end{align*}
    In this case, we have $M_S=0$ for any proper subset $S$ of $[n]$, so Condition~\eqref{eq:strong-assump-main} holds trivially.
    By \cref{thm:main-0}, the mixing time of the Glauber dynamics for the product distribution $\mu$ is $O_\eps(n \log n)$, assuming that each $p_i \in (\eps, 1-\eps)$.
    The dependence on $\eps$ arises from the bound on $1/\mu_{\min}$, which is standard in transferring from modified log-Sobolev inequalities to mixing time bounds.
\end{example}

\begin{example}[Independent sets of matroids]
    Let $\MM = ([n], \II)$ be a matroid  with independent sets $\II\subseteq \set{0,1}^n$. Consider the distribution $\mu$ over independent sets in $\II$ weighted by $\lambda > 0$, given by
    \begin{align*}
        \mu(S) \propto \lambda^{|S|}, \quad \forall S \in \II.
    \end{align*}
    For any non-maximal $S \in \II$ and distinct $i,j \in V_S$, 
    where $V_S = \{i \in [n] \setminus S: S\cup \{i\} \in \II\}$, 
    we have
    \begin{align*}
        M_S(i,j) = \frac{\mu^S(\{i,j\}) \mu^S(\emptyset)}{\mu^S(\{i\}) \mu^S(\{j\})} - 1 
        = - \ind \left\{ \{i,j\} \notin \II^S \right\}.
    \end{align*}
    If we define the contracted matroid $\MM/S = (V_S, \II^S)$ and its rank-$2$ truncation $(\MM/S)_{\leq 2}$ (including only independent sets of size at most 2), then $-M_S$ corresponds to the adjacency matrix of the dependency graph for $(\MM/S)_{\leq 2}$.
    
    A standard fact of matroids is that the dependency graph of any rank-$2$ matroid without trivial elements is a disjoint union of cliques. Thus, there exists a partition of the ground set such that two distinct elements form an independent set iff they belong to different clusters. Hence, we have

    \begin{align*}
        M_S = I - 
        \begin{pmatrix}
            J_1 & & & \\
            & J_2 & & \\
            & & \ddots & \\
            & & & J_k
        \end{pmatrix}
        \preceq I,
    \end{align*}
    where $J_1,\dots,J_k$ are all-one matrices (of possibly distinct sizes), corresponding to the partition of the ground set.
    Applying \cref{thm:main-0}, we conclude that the mixing time of the Glauber dynamics for $\mu$ is $O_\lambda(n \log n)$. This resolves an open question from~\cite[Remark 76]{liu2023spectral}; see \cref{subsec:RC} for applications to the random cluster model on matroids.
\end{example}

\begin{example}[Hardcore model]\label{example:hardcore}
    Let $G=(V,E)$ be a graph, and let $\mu$  be the Gibbs distribution for the hardcore model on $G$ with fugacity $\lambda > 0$, given by
    \begin{align*}
        \mu(S) \propto \lambda^{|S|}, \quad \forall S \in \II(G).
    \end{align*}
    For any non-maximal independent set $S \in \II$, define $V_S = V \setminus (S \cup N(S))$ as the set of vertices that can be added to $S$ to remain as an independent set ($N(S)$ contains neighbors of vertices in $S$). 
    Let $H = G[V_S]$ be the subgraph induced on~$V_S$.
    For any distinct $u,v \in V_S$, we have
    \begin{align*}
        M_S(u,v) = \frac{\mu^S(\{u,v\}) \mu^S(\emptyset)}{\mu^S(\{u\}) \mu^S(\{v\})} - 1 
        = - \ind \left\{ \{u,v\} \in E(H) \right\}.
    \end{align*}
    Thus, $M_S = -A_H$, where $A_H$ is the adjacency matrix of $H$.
    Moreover, by definition, $\*r_S = \lambda \*1$.
    
    Condition~\eqref{eq:weak-assump-main} for $S$ is then equivalent to
    \begin{align*}
        - A_H \preceq \left( 1 + \frac{1-\delta}{\lambda} \right) I
        \quad \Longleftrightarrow \quad
        - \lambda_{\min}(A_H) \le 1 + \frac{1-\delta}{\lambda}
        \quad \Longleftrightarrow \quad
        \lambda \le \frac{1-\delta}{- \lambda_{\min}(A_H) - 1},
    \end{align*}
    where $\lambda_{\min}(\cdot)$ denotes the smallest eigenvalue of a real symmetric matrix.
    Since $\lambda_{\min}(A_H) \ge \lambda_{\min}(A_G)$ for any induced subgraph $H$, Condition~\eqref{eq:weak-assump-main} for all non-maximal $S$ is ensured by the following condition:
    \begin{align*}
        \lambda \le \frac{1-\delta}{- \lambda_{\min}(A_G) - 1}.
    \end{align*}
    This proves \cref{thm:HC-main-lambda-min}. For specific graph families, we obtain the following corollaries:
    \begin{enumerate}
        \item If $G$ has maximum degree $\Delta$, then $\lambda_{\min}(A_G) \ge -\Delta$. 
        Hence, if $\lambda \le \frac{1-\delta}{\Delta-1}$, then by \cref{thm:main-delta}, the Glauber dynamics for $\mu$ has mixing time $O_\lambda(n^2/\delta)$.
        Notably, the condition $\lambda \le \frac{1-\delta}{\Delta-1}$ recovers the Dobrushin's condition for the hardcore model.

        \item If $G$ is a random $\Delta$-regular graph, then with high probability, $- \lambda_{\min}(A_G) = O(\sqrt{\Delta})$. 
        Hence, if $\lambda = O(1/\sqrt{\Delta})$,
        then by \cref{thm:main-delta}, the Glauber dynamics for $\mu$ has mixing time $O_\lambda(n^2/\delta)$.
        This proves \cref{thm:HC-main}. 
        We refer to \cref{subsec:anti-2-spin} for further applications to antiferromagnetic two-spin systems on random regular graphs.
    \end{enumerate}
\end{example}

\begin{example}[Matchings]
    Let $G=(V,E)$ be a \emph{simple} graph, and let $\mu$ be the Gibbs distribution  for the monomer-dimer model of matchings on $G$ with fugacity $\lambda > 0$, given by
    \begin{align*}
        \mu(S) \propto \lambda^{|S|}, \quad \forall S \in \MM(G),
    \end{align*}
    where $\MM(G)$ is the set of all matchings in $G$.
    As matchings correspond to independent sets in the line graph $L(G)$, we deduce from \Cref{example:hardcore} that Condition~\eqref{eq:weak-assump-main} holds for every non-maximal matching $S$ if
    \begin{align*}
        \lambda \le 1-\delta \le \frac{1-\delta}{- \lambda_{\min}(A_{L(G)}) - 1},
    \end{align*}
    where $A_{L(G)}$ is the adjacency matrix of the line graph $L(G)$, and the last inequality follows from a standard fact in algebraic graph theory that the minimum eigenvalue of the line graph of any simple graph satisfies $\lambda_{\min}(A_{L(G)})\ge -2$.
    Therefore, by \cref{thm:main-delta}, if $\lambda \le 1-\delta$, the mixing time of the Glauber dynamics for $\mu$ is at most $O_\lambda( \frac{mn}{\delta} \log n)$, where $n = |V|$ and $m = |E|$.

    We make two additional remarks, which prove \Cref{thm:matching-main}:
    \begin{enumerate}
    \item In \cref{thm:main-delta}, the mixing time bound diverges as $\delta \to 0 $. However, this can be fixed by replacing $1/\delta$ with an alternative bound that remains $O(\sqrt{\Delta})$. In particular, for uniform matchings $\lambda = 1$, this can establish a mixing time of $O( \sqrt{\Delta} mn \log n)$.
    \item The same approach extends naturally to $b$-matchings, where each vertex is incident to at most $b$ edges. For uniform $b$-matchings,
    similar approach gives a mixing time of $O(\min\{\Delta^b, bn\} \cdot bmn \log n)$. See \cref{subsec:Holant} for further applications to Holant problems.
\end{enumerate}
\end{example}

\medskip

We also give new mixing results for the Determinantal Point Process (DPP). 
This is a class of distributions defined as
\begin{align*}
    \mu(S) \propto \left( \mathrm{det}(L_{S,S}) \right)^\alpha, \quad \forall S \subseteq [n],
\end{align*}
where $L \in \R^{n \times n}$ is a symmetric and positive semidefinite matrix, and $\alpha \in [0,1]$.
We refer to \cref{subsec:DPP} for backgrounds on DPP and our results.

\subsection{Proof outline}
\label{subsec:proof-outline}

\begin{figure}[!htbp]
    \centering
    \begin{tikzpicture}[
        block/.style = {rectangle, rounded corners, minimum width=3cm, minimum height=1cm, text centered, text width=4cm, draw=black, fill=orange!30},
        arrow/.style = {thick,->,>=stealth}
        ]
    
        \node[block] (b1) at (0,0) {Conditions \eqref{eq:strong-assump-main}/\eqref{eq:weak-assump-main}};
        \node[block] (b2) at (11,0) {Entropic/Spectral Stability};
        \node[block] (b3) at (11,-2.5) {Entropy/Variance\\Contraction for\\Field Dynamics};
        \node[block] (b4) at (0,-2.5) {Entropy/Variance\\Contraction for\\Glauber Dynamics};

        \draw[arrow] (b1) -- (b2)
        node[midway, above] {Trickle-down for Field Dynamics}
        node[midway, below] {\cref{thm:cor-bound}};
        \draw[arrow] (b2) -- (b3)
        node[midway, left] {\cite{chen2022localization,chen2024rapid}}
        node[midway, right] {\cref{thm:mixing-bound}};
        \draw[arrow] (b3) -- (b4)
        node[midway, above] {New Comparison Lemma}
        node[midway, below] {\cref{lem:limit-field}};
        
    \end{tikzpicture}

    \caption{Proof Outline}
    \label{fig}
\end{figure}


The key components of our proof approach is summarized in \cref{fig}.
Our main technical contribution is the \emph{trickle-down theorem for field dynamics}.
The trickle-down equation for field dynamics establishes an upper bound on the maximum eigenvalue of the correlation matrix of the target distribution, which is an equivalent formulation of the spectral stability property. This is achieved by analyzing the pairwise dependency matrix and applying Conditions \eqref{eq:strong-assump-main}/\eqref{eq:weak-assump-main}.
With the eigenvalue bounds, we deduce rapid mixing of the field dynamics from spectral/entropic stability by recent works \cite{chen2022localization,chen2024rapid}. 
Finally, we establish a new comparison lemma that, roughly speaking, shows that the field dynamics converges to the continuous-time Glauber dynamics in an appropriate sense, thereby allowing us to bound the mixing time of Glauber dynamics.

Let $\XX \subseteq \set{0,1}^n$ be a non-empty downward closed set family over the ground set $[n]$, and let $\mu$ be a distribution fully supported on $\XX$.
We use $\*m(\mu) := \E[S \sim \mu]{\*1_S}$ to denote the mean vector of distribution $\mu$, and let $\Cov(\mu) := \E[S\sim \mu]{(\*1_S - \*m(\mu))(\*1_S - \*m(\mu))^\intercal}$ denote the covariance matrix of $\mu$, where $\*1_S \in \set{0,1}^n$ is the indicator vector for the set $S$. The mean vector $\*m(\mu^{S}) \in \mathbb{R}^{[n] \setminus S}$ and covariance matrix $\mathrm{Cov}(\mu^{S}) \in \mathbb{R}^{([n] \setminus S) \times ([n] \setminus S)}$ for conditional distribution $\mu^S$ are defined accordingly.


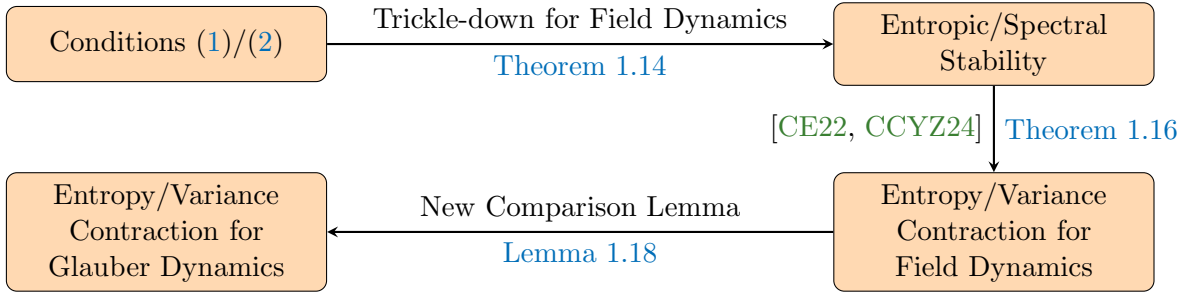
\begin{theorem}[Trickle-down theorem for field dynamics]\label{thm:cor-bound}
    Let $\XX \subseteq \set{0,1}^n$ be a non-empty downward closed set family, and let $\mu$ be a distribution fully supported on $\XX$. 
    Let $\delta \in [0,1]$.
    If for each non-maximal $S \in \XX$ it holds
    \begin{align}\label{eq:weak-assump-trickle}
        M_S \preceq I + (1-\delta)\diag(\*r_S)^{-1}, 
    \end{align}
    then for any non-maximal $S \in \XX$ and $\theta \in (0,1)$, it holds
    \begin{align}\label{eq:main-cov}
        \Cov\left( (1-\theta) * \mu^S \right) \preceq \frac{1}{1-(1-\delta)(1-\theta)} \diag\left( \*m\left( (1-\theta) * \mu^S \right) \right).
    \end{align}
\end{theorem}

Instead of directly proving the mixing time bounds for the Glauber dynamics, we consider a recently introduced Markov chain called \emph{field dynamics} \cite{CFYZ2024rapid}. 

\begin{definition}\label{def:field-dynamics}
    Let $\XX \subseteq \{0,1\}^n$ be a non-empty downward closed set family, and let $\mu$ be a distribution fully supported on $\XX$.
    Let $\theta \in (0,1)$ be a fixed parameter. 
    In each transition step, the \emph{field dynamics} $P^{\mathrm{FD}}_{1 \leftrightarrow \theta} = P_{1 \to \theta} Q_{\theta \to 1}$ updates the current state $X \in \XX$ to a new $X'$ according to the following rules:
    \begin{enumerate}
    \item (Down step $P_{1 \to \theta}$) Independently remove each element $e \in X$ with probability $1-\theta$;
    \item (Up step $Q_{\theta \to 1}$) Sample $X'$ from the distribution $(1-\theta) * \mu^X$.
    \end{enumerate}
\end{definition}

\begin{theorem}[\cite{chen2022localization,chen2024rapid}]\label{thm:mixing-bound}
     Let $\XX \subseteq \{0,1\}^n$ be a non-empty downward closed set family, and let $\mu$ be a distribution fully supported on $\XX$. 
     \begin{enumerate}
     \item If there exists a rate function $C:(0,1) \to \mathbb{R}_{>0}$ such that for any non-maximal $S \in \XX$ and $\theta \in (0,1)$,
     \begin{align}\label{eq:require-spectral}
        \mathrm{Cov}\tp{(1-\theta) * \mu^S} \preceq C(\theta) \cdot \diag\tp{\*m((1-\theta) * \mu^S)},
     \end{align}
     then the Poincar\'{e} inequality for field dynamics $P^{\mathrm{FD}}_{1 \leftrightarrow \theta}$ holds with constant
     \begin{align*}
        \lambda(P^{\mathrm{FD}}_{1 \leftrightarrow \theta}) \ge \exp\tp{-\int_0^\theta \frac{C(\eta)}{1-\eta} \dif \eta}.
     \end{align*}
     \item 
     If for all $\*\lambda \in \mathbb{R}^n_{>0}$ and non-maximal $S \in \XX$, it holds
     \begin{align}\label{eq:require-entropy}
     \mathrm{Cov}(\*\lambda * \mu^S) \preceq \diag(\*m(\*\lambda * \mu^S)),
     \end{align}
     then the modified log-Sobolev inequality for field dynamics $P^{\mathrm{FD}}_{1 \leftrightarrow \theta}$ holds with constant
     \begin{align*}
        \rho_0(P^{\mathrm{FD}}_{1 \leftrightarrow \theta}) \ge \exp\tp{-\int_0^\theta \frac{1}{1-\eta} \dif \eta} = 1-\theta.
     \end{align*}
     \end{enumerate}
\end{theorem}
\begin{remark}
    In our work, we establish an equivalent formulation of spectral stability—an upper bound on the correlation matrix—which is more convenient for analysis within localization schemes.
    We note that the conditions in~\eqref{eq:require-spectral} and~\eqref{eq:require-entropy} are indeed sufficient for spectral stability and entropic stability. For a detailed discussion, we refer the reader to~\Cref{subsec:stability}.
\end{remark}

The final ingredient in our proof is a new comparison lemma that establishes a connection between the mixing times of Glauber dynamics and field dynamics.

\begin{lemma}[Comparison lemma]\label{lem:limit-field}
Let $\XX \subseteq \{0,1\}^n$ be a non-empty downward closed set family, and let $\mu$ be a distribution fully supported on $\XX$. 
Suppose there exist a rate function $\delta: (0,1] \to [0,1]$ and a constant $\delta^\star > 0$ such that the following conditions hold:
\begin{enumerate}
    \item The field dynamics $P^{\mathrm{FD}}_{1 \leftrightarrow 1-\epsilon}$ on $\mu$ satisfies a Poincar\'e inequality with constant $\delta(\epsilon) > 0$, i.e.,
    \begin{align*}
    \forall \epsilon \in (0,1] \text{ and } f \in \mathbb{R}^\Omega, \quad \+E_{P^{\mathrm{FD}}_{1 \leftrightarrow 1-\epsilon}}(f,f) \ge \delta(\epsilon) \cdot \Var[\mu]{f}. 
    \end{align*}
    \item $\lim_{\epsilon \to 0^+} \frac{\delta(\epsilon)}{\epsilon} = \delta^\star$.
\end{enumerate}
Then, the Glauber dynamics $P$ on $\mu$ satisfies a Poincar\'e inequality with constant $\frac{\delta^\star}{(1+r_{\max})n}$, i.e.,
\begin{align*}
    \forall f \in \mathbb{R}^{\Omega}, \quad \+E_{P}(f,f) \ge \frac{\delta^\star}{(1+r_{\max})n} \cdot \Var[\mu]{f}.
\end{align*}
The same result holds if the Poincar\'e inequalities are replaced by modified log-Sobolev inequalities.
\end{lemma}
    Previous comparison arguments in \cite{CFYZ2024rapid, AJKPV2022entropic, CFYZ22optimal} are based on the annealing framework for  Markov chains (see \cite{chen2022localization}), which establishes the rapid mixing of Glauber dynamics for $\mu$ using rapid mixing of field dynamics together with Glauber dynamics at a lower field (i.e., $\lambda * \mu$ for some $\lambda \in (0,1)$). In contrast, our new comparison lemma proves the rapid mixing property only using the field dynamics, without the need to consider Glauber dynamics at a lower field.

\begin{proof}[Proof of~\Cref{thm:main-0}]
Fix $\*\lambda \in \mathbb{R}_{>0}^n$ and non-maximal $S \in \XX$. By~\Cref{thm:cor-bound},~\Cref{remark:invariant} and~\eqref{eq:strong-assump-main}, we have
\begin{align*}
\mathrm{Cov}(\*\lambda * \mu^S) \preceq \diag(\*m(\*\lambda * \mu^S)).
\end{align*}
By~\Cref{thm:mixing-bound,lem:limit-field}, the modified log-Sobolev constant satisfies
\begin{align*}
    \rho_0(P_\mathsc{gd}) \ge \frac{1}{(1+r_{\max}) n}.
\end{align*}
The mixing time bound then follows.
\end{proof}

\begin{proof}[Proof of~\Cref{thm:main-delta}]
Fix non-maximal $S \in \XX$. By~\Cref{thm:cor-bound,thm:mixing-bound}, the Poincar\'e inequality for the field dynamics holds with constant
\begin{align*}
    \lambda(P_{1 \leftrightarrow \theta}^{\mathrm{FD}}) \ge \exp\tp{-\int_0^{\theta} \frac{1}{(1-\eta)(1-(1-\delta)(1-\eta))} \dif \eta} = \frac{(1-\theta)\delta}{1-(1-\theta)(1-\delta)}.
\end{align*}
By~\Cref{lem:limit-field}, this implies a Poincar\'e inequality for the Glauber dynamics with constant 
\begin{align}\label{eq:GD-poincare-1}
\lambda(P_\mathsc{gd}) \ge \frac{\delta}{(1+r_{\max}) n}.
\end{align}
The mixing time bound then follows.
\end{proof}


\section{Preliminaries}
\label{sec:preliminaries}


\subsection{Basics of Markov chains}
Let $\Omega$ be a finite state space, and $(X_t)_{t \ge 0}$ be a Markov chain with transition matrix $P \in \mathbb{R}_{\ge 0}^{\Omega \times \Omega}$. The Markov chain $(X_t)_{t \ge 0}$ is \emph{irreducible} if for any $x,y \in \Omega$, there exists $t \ge 0$ satisfying $P^t(x,y) > 0$, and is \emph{aperiodic} if for any $x \in \Omega$, $\mathrm{gcd} \set{t \ge 1 \mid P^t(x,x) > 0} = 1$.
By the fundamental theorem of Markov chain, an irreducible and aperiodic Markov chain $(X_t)_{t \ge 0}$ converges to its unique \emph{stationary distribution} $\mu$, i.e. a distribution $\mu$ satisfying $\mu P = \mu$. We may use the transition matrix $P$ to refer to the Markov chain $(X_t)_{t \ge 0}$.

Let $P$ be an irreducible and aperiodic Markov chain with stationary distribution $\mu$. The mixing time of the Markov chain $P$ is given by
\begin{align*}
T_{\mathrm{mix}}(P) = \max_{x \in \Omega} \min \set{t \ge 0 \mid d_{\mathrm{TV}}(P^t(x,\cdot),\mu) < \frac{1}{4}},
\end{align*}
where $d_{\mathrm{TV}}(\mu,\nu) = \frac{1}{2}\sum_{x \in \Omega} \abs{\mu(x)-\nu(x)}$ denotes the total variation distance.

\subsection{Functional inequalities and mixing time of Markov chains}
Let $\mu$ be a distribution over $\Omega \subseteq \{0,1\}^{n}$, and $f \in \mathbb{R}^\Omega$ and $g \in \mathbb{R}_{>0}^{\Omega}$ be functions. The \emph{expectation} $\E[\mu]{f}$ is defined as $\E[\mu]{f} = \sum_{x \in \Omega} \mu(x) f(x)$. The \emph{variance} $\Var[\mu]{f}$ and \emph{entropy} $\Ent[\mu]{f}$ is
\begin{align*}
    \Var[\mu]{f} = \mathbb{E}_\mu\left[{f^2}\right] - \tp{\mathbb{E}_\mu\left[{f}\right]}^2,
\end{align*}
and
\begin{align*}
\Ent[\mu]{g} = \mathbb{E}_\mu\left[g \log g\right] - \mathbb{E}_\mu\left[g\right] \log \mathbb{E}_\mu\left[g\right].
\end{align*}
The Poincar\'e inequality and modified log-Sobolev inequality are the important tool in analyzing the mixing time of Markov chains.
\begin{definition}  
    Let $(X_t)_{t \ge 0}$ be an irreducible and aperiodic Markov chain on a finite state space $\Omega$ with transition matrix $P \in \mathbb{R}^{\Omega \times \Omega}$, and $\mu$ be the unique stationary distribution. 
    \begin{itemize}
        \item We say the Poincar\'e inequality holds with constant $\lambda > 0$ if 
        \begin{align*}
            \forall f \in \mathbb{R}^{\Omega}, \quad \lambda \cdot \Var[\mu]{f} \le \+E_P(f,f),
        \end{align*}
        where $\+E_P(f,g) := \inner{f}{(I-P)f}_\mu$ denotes the Dirichlet form and $\inner{f}{g}_{\mu} := \sum_{x \in \Omega} \mu(x) f(x) g(x)$ denotes the inner product of $f$ and $g$.
        \item We say the modified log-Sobolev inequality holds with constant $\rho_0 > 0$ if 
        \begin{align*}
            \forall f \in \mathbb{R}_{>0}^\Omega,\quad \rho_0 \cdot \Ent[\mu]{f} \le \+E_P(f,\log f).
        \end{align*}
    \end{itemize}
    The following results on the mixing time of Glauber dynamics via functional inequalities are known.
    \begin{lemma}
        Let $(X_t)_{t \ge 0}$ be an irreducible and aperiodic Markov chain on a finite state space $\Omega$ with transition matrix $P \in \mathbb{R}^{\Omega \times \Omega}$, and $\mu$ be the unique stationary distribution. 
        \begin{itemize} 
            \item If the Poincar\'e inequality holds with constant $\lambda > 0$, then the mixing time satisfies
            \begin{align*}  
                T_{\mathrm{mix}} \le \frac{1}{\lambda} \cdot \log \frac{4}{\mu_{\min}},
            \end{align*}
            where $\mu_{\min} := \min_{x \in \Omega} \mu(x)$.
            \item If the modified log-Sobolev inequality holds with constant $\rho_0 > 0$, then the mixing time satisfies
            \begin{align*}
                T_{\mathrm{mix}} \le \frac{1}{\rho_0} \cdot \tp{\log \log \frac{1}{\mu_{\min}} + 4}.
            \end{align*}
        \end{itemize}
    \end{lemma}
\end{definition}

\subsection{Field dynamics}

Introduced in~\cite{CFYZ2024rapid}, field dynamics has proven to be an effective tool for analyzing the mixing time of Markov chains~\cite{AJKPV2022entropic,CFYZ22optimal,chen2022localization,chen2024rapid}. Building on the recent work~\cite{chen2024rapid}, we interpret the field dynamics as a continuous-time down-up walk.
    

\begin{definition} \label{def:continuous-time-down-up-walk}
    Let $\mu$ be a distribution over $\XX \subseteq \set{0,1}^n$, and $0 \le \eta \le \theta \le 1$ be fixed parameters.
    \begin{itemize}
    \item The \emph{continuous-time down process} $(X_t)_{t \in [0,1]}$ with initial distribution $X_0 \sim \mu$ is a Markov process constructed as follows. Initially, each site $i \in [n]$ is assigned with an independent random variable $r_i$ uniformly distributed on $[0,1]$. For $t \in [0,1]$, $X_t$ is given by 
    \begin{align*}
    X_t = \{i \in X_0 \mid r_i \ge t\}.
    \end{align*}

    \item The \emph{continuous-time up process} $(Y_t)_{t \in [0,1]}$ is the time-reversal of $(X_t)_{t \in [0,1]}$. In particular, the law of $X_t$ is identical to the law of $Y_{1-t}$ for all $t \in [0,1]$.
    
    \item The Markov kernels $P_{\theta \to \eta}(S,\cdot)$ and $Q_{\eta \to \theta}(S,\cdot)$ are defined as follows:
    \begin{align*}
        P_{\theta \to \eta}(S,\cdot) := \Pr[]{X_{1-\eta} = \cdot \mid X_{1-\theta} = S} \quad \text{and} \quad Q_{\eta \to \theta}(S,\cdot) := \Pr[]{Y_\theta = \cdot \mid Y_{\eta} = S}.
    \end{align*}
    The \emph{continuous-time down-up walk}, known as the \emph{field dynamics}, is defined as 
    \begin{align*}
    P^{\mathrm{FD}}_{1\leftrightarrow \theta} := P_{1 \to \theta} Q_{\theta \to 1}.
    \end{align*}
    \end{itemize}
\end{definition}

\section{Spectral stability, entropic stability, and log-concavity}
\label{subsec:stability}

In recent years, the concepts of spectral and entropic independence/stability have played a central role in establishing mixing time bounds for Markov chains, including Glauber dynamics, the down-up walk, field dynamics, and the proximal sampler. 
In this section, we review a notion of spectral/entropic stability \emph{with respect to the field dynamics} and examine its connection to the mixing time of Markov chains.

\subsection{Spectral stability}

The notion of spectral stability was first introduced in~\cite{chen2022localization} as a generalization of spectral independence to arbitrary \emph{localization schemes}. 
In this work, we focus specifically on the field dynamics (referred to as \emph{negative-fields localization} in~\cite{chen2022localization}) and define spectral/entropic stability with respect to this particular localization scheme. 
Our definitions follow those in~\cite{chen2024rapid}, which are slightly different from the formulation in~\cite{chen2022localization}.

\begin{definition}[Spectral Stability w.r.t.~Field Dynamics, \cite{chen2024rapid}]
    Let $\mu$ be a distribution over $\XX \subseteq \{0,1\}^n$, and let $C:[0,1) \to \mathbb{R}_{\ge 0}$ be a rate function. 
    
    We say that $\mu$ satisfies \emph{spectral stability with rate $C$ with respect to the field dynamics}, if for every $\eta \in [0,1)$, $S \in \XX$ and $f : \XX \to \mathbb{R}$, it holds
    \begin{align}\label{eq:spec-stab-def}
        \lim_{\theta \to \eta^+} \frac{1}{\theta - \eta} \Var[Q_{\eta \to \theta}(S,\cdot)]{Q_{\theta \to 1} f}
        \le\frac{C(\eta)}{1-\eta} \cdot \Var[Q_{\eta \to 1}(S,\cdot)]{f}.
    \end{align}
\end{definition}

Roughly speaking, the notion of spectral stability characterizes the speed of information loss along the continuous-time down process (or, equivalently, the speed of information gain along the continuous-time up process).
Although being abstract in nature, spectral stability directly provides a lower bound on the Poincar\'e constant for field dynamics.

\begin{lemma}[\cite{chen2022localization,chen2024rapid}]
\label{lem:mixing-bound-original}
     Let $\mu$ be a distribution over $\XX \subseteq \{0,1\}^n$ and $C:[0,1) \to \mathbb{R}_{\ge 0}$ be a rate function. 
    If the distribution $\mu$ satisfies spectral stability with rate $C$ with respect to the field dynamics, then the Poincar\'{e} constant for the field dynamics $P^{\mathrm{FD}}_{1 \leftrightarrow \theta}$ satisfies
     \begin{align}
        \lambda(P^{\mathrm{FD}}_{1 \leftrightarrow \theta}) \ge \exp\tp{-\int_0^\theta \frac{C(\eta)}{1-\eta} \dif \eta}.
     \end{align}
\end{lemma}

The following proposition, implicitly from \cite{chen2024rapid}, provides a much more intuitive interpretation of the spectral stability with respect to field dynamics.
These conditions can effectively serve as alternative definitions of the notion. 
Notably, the first part of~\Cref{thm:mixing-bound} follows directly from~\Cref{lem:mixing-bound-original} and \cref{prop:equiv-SI-stability}.


\begin{proposition}[Equivalent Definitions of Spectral Stability w.r.t.~Field Dynamics]
\label{prop:equiv-SI-stability}
    Let $\mu$ be a distribution over $\XX \subseteq \{0,1\}^n$ and $C:[0,1) \to \mathbb{R}_{\ge 0}$ be a rate function. 
    Then, the following conditions are equivalent:

    \begin{enumerate}
        \item\label{item:spectral-1} The distribution $\mu$ satisfies spectral stability with rate $C$ with respect to the field dynamics; 
        
        \item\label{item:spectral-2} For any $\eta \in [0,1)$ and $S \in \XX$, it holds
        \begin{align}
            \sum_{i \in [n] \setminus S} p_i \tp{\frac{q_i}{p_i}-1}^2 \le C(\eta) \cdot D_{\chi^2}(\nu \parallel (1-\eta) * \mu^S),
            \qquad \text{for all } \nu \ll \mu^S
        \end{align}
        where $\*p = \*m((1-\eta) * \mu^S)$ and $\*q = \*m(\nu)$ are the mean vectors;

        \item\label{item:spectral-3} For any $\eta \in [0,1)$ and $S \in \XX$, it holds
        \begin{align}\label{eq:SS-Cov-DMean}
            \mathrm{Cov}((1-\eta)*\mu^S) \preceq C(\eta) \cdot \diag\tp{\*m((1-\eta)*\mu^S)}.
        \end{align}
    \end{enumerate} 
\end{proposition}

The proof of~\Cref{prop:equiv-SI-stability} is included in \cref{appendix-stab-proof} for completeness.

\begin{remark}[Comparison of Spectral Stability w.r.t.~Field Dynamics and Spectral Independence]
\label{rmk:SI-SS}
    Since its introduction \cite{ALO2021spectral}, spectral independence has become a powerful tool for analyzing the Glauber dynamics.
    In the framework of localization schemes \cite{chen2022localization},
    the notion of spectral independence can be viewed as spectral stability with respect to the uniform-block dynamics (termed as \emph{coordinate-by-coordinate localization} in \cite{chen2022localization}). 
    When no pinning is applied, it can be defined formally by
    \begin{align}\label{eq:rmk-SI}
        \mathrm{Cov}(\mu) \preceq C \cdot \diag\tp{\mathrm{Var}(\mu)},
    \end{align}
    where $C$ is some constant and $\diag\tp{\mathrm{Var}(\mu)}$ is a diagonal matrix of variances.
    Equivalently, and more commonly, spectral independence is defined by
    \begin{align*}
        \lambda_{\max}(\Psi) \le C,
    \end{align*}
    where $\Psi \in \R^{n \times n}$ is the influence matrix defined as
    \begin{align*}
        \Psi(i,j) = \mathbb{P}_{S \sim \mu}(j \in S \mid i \in S) - \mathbb{P}_{S \sim \mu}(j \in S \mid i \notin S),
        \quad \forall i,j \in [n].
    \end{align*}
    
    There is a delicate difference between the usual spectral independence in literature and the spectral stability in this paper with respect to field dynamics. 
    When no pinning is applied (i.e., $\eta = 0$ and $S=\emptyset$), spectral stability \cref{eq:SS-Cov-DMean} can be written as
    \begin{align}\label{eq:rmk-SS}
        \mathrm{Cov}(\mu) \preceq C \cdot \diag\tp{\*m(\mu)},
    \end{align}
    where $C$ is some constant and $\diag\tp{\*m(\mu)}$ is a diagonal matrix of the mean vector.
    It is easy to see that this is further equivalent to the condition 
    \begin{align*}
        \lambda_{\max}(\tilde{\Psi}) \le C,
    \end{align*}
    where $\tilde{\Psi} \in \R^{n \times n}$ is a reweighted influence matrix, called the \emph{correlation matrix} in \cite[Definition 15]{AASV21sector}, defined as
    \begin{align*}
        \tilde{\Psi}(i,j) = \mathbb{P}_{S \sim \mu}(j \in S \mid i \in S) - \mathbb{P}_{S \sim \mu}(j \in S), 
        \quad \forall i,j \in [n].
    \end{align*}

    In the framework of simplicial complexes, the notion of spectral independence was also studied for distributions supported on the maximal faces of a pure simplicial complex, which also takes the form of \cref{eq:rmk-SS} \cite{AASV21sector}. In fact, spectral independence \cref{eq:rmk-SI} for a distribution on the product space $\{0,1\}^n$ is derived by considering the pure simplicial complex on the ground set $[n] \times \{0,1\}$, where every configuration in $\{0,1\}^n$ is viewed as an $n$-subset of the ground set $[n] \times \{0,1\}$, and the Glauber dynamics is equivalent to the down-up walk in this complex \cite{ALO2021spectral}.
    In contrast, our spectral stability for field dynamics \cref{eq:rmk-SS} is defined for a distribution supported over \emph{all} faces of a simplicial complex that is not necessarily pure, and hence generalizes the previous setting of maximal faces in pure simplicial complexes.

    We remark that since $\diag\tp{\mathrm{Var}(\mu)} \preceq \diag\tp{\*m(\mu)}$ for any distribution $\mu$ over $\{0,1\}^n$, \cref{eq:rmk-SI} implies \cref{eq:rmk-SS}; namely, spectral independence is a slightly stronger condition than spectral stability with respect to field dynamics.
    In most previous works, the two definitions do not make a significant difference in applications, since most of times one only aims to show an $O(1)$ upper bound for either notion. However, in this work we need precisely the spectral stability with respect to the field dynamics to perform the trickle-down argument. Furthermore, spectral stability is the correct form to characterize log-concavity of the generating polynomials, see \cref{subsec:log-concave}.
\end{remark}

\subsection{Entropic stability}
By changing the variance in spectral stability to entropy, we obtain the notion of entropic stability.

\begin{definition}[Entropic Stability w.r.t.~Field Dynamics,~\cite{chen2024rapid}]
    Let $\mu$ be a distribution over $\XX \subseteq \{0,1\}^n$, and let $C:[0,1) \to \mathbb{R}_{\ge 0}$ be a rate function. 
    
    The distribution $\mu$ satisfies \emph{entropic stability with rate $C$ with respect to the field dynamics}, if for any $\eta \in [0,1)$ and $S \in \XX$, we have
    \begin{align}\label{eq:ent-stab-def}
        \lim_{\theta \to \eta^+} \frac{1}{\theta - \eta} \Ent[Q_{\eta \to \theta}(S,\cdot)]{Q_{\theta \to 1} f} \le \frac{C(\eta)}{1-\eta} \cdot \Ent[Q_{\eta \to 1}(S,\cdot)]{f},
        \qquad \text{for all } f: \XX \to \mathbb{R}_{\ge 0}.
    \end{align}
\end{definition}

Just as spectral stability implies the Poincar\'{e} inequality for the field dynamics (see \cref{lem:mixing-bound-original}), entropic stability implies the modified log-Sobolev inequality for the field dynamics, as stated in the lemma below. 

\begin{lemma}[\cite{chen2022localization,chen2024rapid}]\label{lem:mixing-bound-original-entropic}
     Let $\mu$ be a distribution over $\XX \subseteq \{0,1\}^n$ and $C:[0,1) \to \mathbb{R}_{\ge 0}$ be a rate function. 
    If the distribution $\mu$ satisfies entropic stability with rate $C$ with respect to the field dynamics, then the modified log-Sobolev constant for the field dynamics $P^{\mathrm{FD}}_{1 \leftrightarrow \theta}$ satisfies
     \begin{align}\label{eq:require-entropic}
        \rho_0(P^{\mathrm{FD}}_{1 \leftrightarrow \theta}) \ge \exp\tp{-\int_0^\theta \frac{C(\eta)}{1-\eta} \dif \eta}.
     \end{align}
\end{lemma}

The second part of~\Cref{thm:mixing-bound} then follows from~\Cref{lem:mixing-bound-original-entropic} and the following proposition.

\begin{proposition}[Equivalent Definitions \& Sufficient Condition of Entropic Stability w.r.t.~Field Dynamics]
\label{prop:entropic-sufficient-condition}
    Let $\mu$ be a distribution over $\XX \subseteq \{0,1\}^n$ and $C:[0,1) \to \mathbb{R}_{\ge 0}$ be a rate function.
    Then the following are equivalent:

    \begin{enumerate}
        \item\label{item:entropic-1} The distribution $\mu$ satisfies entropic stability with rate $C$ with respect to the field dynamics;

        \item\label{item:entropic-2} For any $\eta \in [0,1)$ and $S \in \XX$, we have
        \begin{align}
            \sum_{i \in [n] \setminus S} q_i \log \frac{q_i}{p_i} - (q_i - p_i) \le C(\eta) \cdot D_{\mathrm{KL}}(\nu \parallel (1-\eta)*\mu^S), 
            \qquad \text{for all } \nu \ll (1-\eta)*\mu^S
        \end{align}
        where $\*p = \*m((1-\eta)*\mu^S)$ and $\*q = \*m(\nu)$ are the mean vectors;
    \end{enumerate}
    Furthermore, if there is a constant $C^\star$ such that the following holds:
    \begin{enumerate}
       \setcounter{enumi}{2}
        \item \label{item:entropic-3} For any $\*\lambda \in \mathbb{R}_{>0}^{n}$ and $S \in \XX$, we have
        \begin{align}
            \mathrm{Cov}(\*\lambda*\mu^S) \preceq C^\star \cdot \diag\tp{\*m(\*\lambda*\mu^S)};
        \end{align}
    \end{enumerate}
    then the distribution $\mu$ satisfies entropic stability with constant rate $C(\eta) \equiv C^\star$ with respect to the field dynamics.
\end{proposition}

The proof of~\Cref{prop:entropic-sufficient-condition} can be found in \cref{appendix-stab-proof} for completeness.

\subsection{Log-concavity}
\label{subsec:log-concave}

Our \cref{thm:cor-bound} shows that 
Condition \eqref{eq:strong-assump-main} implies spectral and entropic stability with constant rate $1$ with respect to field dynamics.
In fact, the converse is also true, and all these properties are equivalent to the log-concavity of the generating polynomial.
We summarize these equivalences in the proposition below. 

For a distribution $\mu$ supported on $\XX \subseteq \{0,1\}^n$, define its \emph{generating polynomial} $g_\mu: \mathbb{R}^{[n]} \to \R$ to be
\begin{align*}
    g_\mu(\*\lambda) = \sum_{S \in \XX} \mu(S) \prod_{i \in S} \lambda_i.
\end{align*}

\begin{proposition}[Equivalent Definitions of Log-concavity]
\label{prop:log-concave}
    Let $\XX \subseteq \{0,1\}^n$ be a non-empty downward closed set family, and let $\mu$ be a distribution fully supported on $\XX$.
    Then all of the following are equivalent:
    \begin{enumerate}
        \item The distribution $\mu$ satisfies \emph{Condition \eqref{eq:strong-assump-main}}; i.e., for any non-maximal $S \in \XX$, we have $M_S \preceq I$;

        \item For any $\*\lambda \in \mathbb{R}_{>0}^{n}$, the distribution $\*\lambda*\mu$ satisfies \emph{spectral stability} with constant rate $1$ with respect to the field dynamics;

        \item For any $\*\lambda \in \mathbb{R}_{>0}^{n}$, the distribution $\*\lambda*\mu$ satisfies \emph{entropic stability} with constant rate $1$ with respect to the field dynamics;

        \item The generating polynomial $g_\mu$ is \emph{strongly log-concave}; i.e., for any $S \in \XX$, its partial derivative $\partial^S g_\mu = g_{\mu^S}$ is log-concave on $\mathbb{R}_{>0}^{n}$.
    \end{enumerate}
\end{proposition}

Though our main result \cref{thm:main-0} does not rely on \cref{prop:log-concave}, we believe these equivalences are interesting by themselves. 
A homogeneous version of \cref{prop:log-concave} (i.e., $\mu$ is supported on subsets of a fixed size and $g_\mu$ is a homogeneous polynomial) was already known by seminal works \cite{anari2019logconcaveII,branden2020lorentzian}.
One way to establish \cref{prop:log-concave} is to utilize  
the equivalence of strong/complete log-concavity of any polynomial $p$ and a suitable homogenization $p_{\mathrm{homo}}$ of it, see \cite{anari2019logconcaveII,branden2020lorentzian}.
This also provides an alternative way to derive Condition \eqref{eq:strong-assump-main}.
In \cref{appendix-stab-proof}, we give a more direct proof based on facts established in this paper, circumventing the need of homogenization.

\section{Proof of Main Results}

In this section, we prove the main technical results \cref{thm:cor-bound,lem:limit-field} from the introduction.
In \cref{subsec:trickle-down}, we prove \cref{thm:cor-bound} which establishes spectral stability with respect to the field dynamics using the trickle-down equation given in \cite{anari2024trickle}.
Then in \cref{subsec:comparison}, we prove \cref{lem:limit-field}, establishing a simple comparison between the field dynamics and the Glauber dynamics.

\subsection{Trickle-down theorem for field dynamics}
\label{subsec:trickle-down}
The trickle-down theorem for field dynamics (\Cref{thm:cor-bound}) relies on the following trickle-down equation for field dynamics, which is a special case of the trickle-down equation for general linear-tilt localization schemes presented in~\cite{anari2024trickle}.
\begin{proposition}[trickle-down equation for field dynamics,~\cite{anari2024trickle}]\label{prop:trickle-down-eq}
    Let $\mu$ be a distribution supported over $\XX \subseteq \{0,1\}^n$. Let $(X_t)_{t \in [0,1]}$ and $(Y_t)_{t \in [0,1]}$ be the continuous-time down and up processes respectively for the target distribution $\mu$ where $X_0 = Y_1 \sim \mu$. 
    Then for sufficiently small $h > 0$, we have
    \begin{align*}
        \Sigma = \E[]{\mathrm{Cov}(Y_1 \mid Y_h)} + \Sigma \Pi^{-1} \Sigma \cdot h + o(h),
    \end{align*}
    where $\Sigma = \mathrm{Cov}(\mu)$ denotes the covariance matrix and $\Pi = \diag\tp{\*m(\mu)}$ denotes the diagonal matrix of the mean vector of $\mu$.
    We let $\Pi^{-1} := \-{diag}(\Pi_{ii}^{-1})$ with convention $0^{-1} = 0$ (i.e., the pseudoinverse).
\end{proposition}

\begin{proof}
By the total law of covariance, we have
\begin{align*}
    \Sigma = \mathrm{Cov}(Y_1) = \E[]{\mathrm{Cov}(Y_1 \mid Y_h)} + \mathrm{Cov}\tp{\E[]{Y_1 \mid Y_h}}.
\end{align*}
It suffices to show that $\mathrm{Cov}\tp{\E[]{Y_1 \mid Y_h}} = \Sigma \Pi^{-1} \Sigma \cdot h + o(h)$.
By definition, we have
\begin{align} \label{eq:Cov(E[Y1|Yt+h]|Yt)}
  \-{Cov}\tp{\E{Y_1 \mid Y_{h}}}
  &= \E{ (\E{Y_1 \mid Y_{h}} - \E{Y_1})^{\otimes 2}}.
\end{align}
Therefore, we can calculate the law of $Y_{h}$ as follows: for $x\supseteq y$,
\begin{align*}
  \Pr{Y_{h} = x}
  &= \sum_{z\supseteq x} \mu(z) h^{\abs{x}} (1-h)^{\abs{z}-\abs{x}} 
  = \begin{cases}
      o(h), &\abs{x} \geq 2;\\
      (h+o(h)) \; \*m(\mu)_i,  &x = \set{i}; \\
      1 - h \sum_i \*m(\mu)_i + o(h),  &x = \emptyset.
  \end{cases}
\end{align*}
Similarly, for every $y \in [n]$, we have $\E{Y_1 \mid Y_h = \set{y}} = \*m(\mu^{\set{y}}) + o(h)$, where we recall that $\mu^{\set{y}}$ denotes the distribution of $S \sim \mu$ conditioned on $y \in S$.
This implies
\begin{align*}
    \-{Cov}\tp{\E{Y_1 \mid Y_{h}}} 
    &= \sum_i h\; \*m(\mu)_i \tp{\*m(\mu^{\set{i}}) -\*m(\mu)}^{\otimes 2} + o(h).
\end{align*}
Hence, for every $j, k \in [n],$ we have
\begin{align*}
    \-{Cov}\tp{\E{Y_1 \mid Y_{h}}}_{ij} &= \sum_i \*m(\mu)_i (\*m(\mu^{\set{i}})_j - \*m(\mu)_j) (\*m(\mu^{\set{i}})_k - \*m(\mu)_k) + o(h)\\
    &= \sum_i \*m(\mu)_i^{-1} \-{Cov}(\mu)_{ji}\-{Cov}(\mu)_{ik} + o(h).
\end{align*}
Therefore, we conclude that
\begin{align*}
\mathrm{Cov}\tp{\E[]{Y_1 \mid Y_h}} 
    &=\Sigma \Pi^{-1} \Sigma \cdot h + o(h). \qedhere
\end{align*}
\end{proof}

We first outline the core idea in the proof of~\Cref{thm:cor-bound}. Let $\mu$ be a distribution over $\XX \subseteq \{0,1\}^n$. Define $f:[0,1] \to \mathbb{R}_{\ge 1}$ as follows:
\begin{align}\label{eq:def-f}
    \forall \lambda \in [0,1],\quad f(\lambda) = \begin{cases}
        \sup_{S \in \XX} \norm{ \diag\tp{\*m(\lambda * \mu^S)}^{-1} \Cov(\lambda * \mu^S)}_2 & \lambda > 0;\\
        1 & \lambda = 0.
    \end{cases}
\end{align}
We remark that $f(\lambda)$ is well-defined, as the $\ell_2$-norm of the correlation matrix has a crude bound~$n$. Furthermore, $f$ is continuous on the interval $[0,1]$,
since $f$ is the supremum of finitely many continuous functions and $\lim_{\lambda\to 0^+} f(\lambda) =1 $, see \cref{prop:ode-boundary,eq:small-lambda}.

Fix $\lambda \in (0,1]$ and a feasible subset $S \subseteq [n]$.
Let $(X_t)_{t \in [0,1]}$ and $(Y_t)_{t \in [0,1]}$ be the continuous-time down and up walks respectively with target distribution $\lambda * \mu^S$ where $X_0 = Y_1 \sim \lambda * \mu^S$.
Note that
the distribution $Y_1 \mid Y_h$ follows the law of $Y_h \cup Z$, where $Z \sim (1-h) \lambda *\mu^{S \cup Y_h}$. 
Hence, the covariance matrix $\mathrm{Cov}(Y_1 \mid Y_h)$ is obtained by augmenting $\mathrm{Cov}((1-h)\lambda * \mu^{S \cup Y_h})$ with additional rows and columns of zeroes, and the mean vector $\*m(Y_1 \mid Y_h)$ is obtained by appending ones to $\*m((1-h) \lambda * \mu^{S \cup Y_h})$.
By~\Cref{prop:trickle-down-eq}, it holds
\begin{align}\label{eq:ineq-Sigma}
    \nonumber \Sigma &= \E[]{\mathrm{Cov}(Y_1 \mid Y_h)} + \Sigma \Pi^{-1}\Sigma \cdot h + o(h)\\
    \nonumber(\text{by~\eqref{eq:def-f}})\quad&\preceq 
    f((1-h) \lambda) \cdot \E[]{\diag \tp{\*m(Y_1 \mid Y_h) - \*1_{Y_h}}} + \Sigma \Pi^{-1} \Sigma \cdot h + o(h)\\
    &\overset{(\star)}{=} (1-h) f((1-h)\lambda) \cdot \Pi + \Sigma \Pi^{-1} \Sigma \cdot h + o(h),
\end{align}
where $\Sigma = \mathrm{Cov}(Y_1) = \mathrm{Cov}(\lambda*\mu^S)$ denotes the covariance matrix, $\Pi=\diag\tp{\*m(\lambda *\mu^S)}$ denotes the diagonal matrix of mean vector, $\*1_{Y_h}$ denotes the indicator vector, and $(\star)$ follows from the law of total expectation and $\Pr[]{i \in Y_h} = h \cdot \Pr[]{i \in Y_1} = h \cdot \Pi_{i,i}$.

Without loss of generality, we assume $\Pi_{i,i} > 0$ for every $i$; otherwise, we can restrict to a smaller ground set and apply the same argument.
From~\eqref{eq:ineq-Sigma}, we have
\begin{align*}
  \Pi^{-1/2}\Sigma \Pi^{-1/2} \preceq (1-h) f((1-h)\lambda) \cdot I + h \cdot (\Pi^{-1/2}\Sigma \Pi^{-1/2})^2 + o(h),
\end{align*}
Let $S$ be the set that achieves supremum in \eqref{eq:def-f}, then
$f(\lambda) = \norm{\Pi^{-1/2} \Sigma}_2 = \norm{\Pi^{-1/2}\Sigma \Pi^{-1/2}}_2$.
Taking $2$-norm at both sides and using its sub-additivity, we have the following constraint on $f$:
\begin{align}\label{eq:f-derivative}
    f(\lambda) \le (1-h) f((1-h)\lambda) + f(\lambda)^2 \cdot h + o(h).
\end{align}
Therefore, for any $\lambda \in (0,1]$, the Dini derivative of $f$ satisfies:
\begin{align}\label{eq:ode-main}
    D^{-}f(\lambda) := \limsup_{h \to 0^+} \frac{f(\lambda) - f((1-h)\lambda)}{h\lambda} \overset{\eqref{eq:f-derivative}}{\le} \limsup_{h \to 0^+} \frac{f(\lambda)^2 - f((1-h)\lambda)}{\lambda} = \frac{f(\lambda)^2-f(\lambda)}{\lambda},
\end{align}
where the last equality follows from the continuity of $f$. 

We now assert that the assumption~\eqref{eq:weak-assump-trickle} characterizes the Dini derivative of $f$ at origin.
\begin{proposition}\label{prop:ode-boundary}
    If the support $\XX$ of $\mu$ is non-empty and downward closed,
    then~\eqref{eq:weak-assump-trickle} implies
    \begin{align}\label{eq:ode-boundary}
    D^+f(0) = \limsup_{\lambda \to 0^+} \frac{f(\lambda) - 1}{\lambda} \le 1-\delta.
\end{align}
\end{proposition}
The proof of~\Cref{prop:ode-boundary} is deferred to the end of this section. 

Constraints~\eqref{eq:ode-main},~\eqref{eq:ode-boundary}, together with trivial restrictions on $f$, form the following differential inequalities:
\begin{align}\label{eq:dif-ineq}
\begin{cases}
D^{-}f(\lambda) \le \frac{f(\lambda)^2 - f(\lambda)}{\lambda}, & \forall \lambda \in (0,1]; \\
D^+ f(0) \le 1-\delta;\\
f(0) = 1 \text{ and }
f(\lambda) > 0, & \forall \lambda \in [0,1].
\end{cases}
\end{align}
As will be shown soon, the solution to the differential inequalities~\eqref{eq:dif-ineq} must satisfy 
\begin{align}\label{eq:f-bound}
    f(\lambda) \le \frac{1}{1-(1-\delta) \lambda}.
\end{align}
We remark that solving the differential inequalities~\eqref{eq:dif-ineq} can be viewed as a modified version of the Gr\"onwall's inequality, and our proof approach can also be understood as a continuous analog of the original trickle-down argument for discrete-time down and up walks \cite{Opp18,anari2024trickle}.

\begin{proof}[Proof of~\Cref{thm:cor-bound}]
From our arguments above, it remains to solve the differential inequalities~\eqref{eq:dif-ineq} and show \eqref{eq:f-bound}.
Let $\alpha = 1-\delta$ and $g:(0,1] \to \mathbb{R}$ be a continuous function satisfying 
\begin{align*}
    \forall \lambda \in (0,1],\quad  f(\lambda) = \frac{1}{1-\alpha \lambda g(\lambda)}.
\end{align*}
By a straightforward calculation, the Dini derivative of $f$ satisfies
\begin{align*}
    D^- f(\lambda) = \frac{\alpha g(\lambda) + \alpha \lambda \cdot D^- g(\lambda)  }{\tp{1-\alpha \lambda g(\lambda) }^2} = \frac{f(\lambda)^2-f(\lambda)}{\lambda} + \frac{\alpha \lambda \cdot D^{-} g(\lambda)}{(1-\alpha \lambda g(\lambda))^2}.
\end{align*}
Together with~\eqref{eq:ode-main}, it implies $D^- g(\lambda) \le 0$ for all $\lambda \in (0,1]$. By~\cite[Remark 2.1]{szarski1965}, this implies $g$ is monotone decreasing on $(0,1]$. Furthermore, the boundary condition $D^+ f(0) \le \alpha$ reads
\begin{align*}
D^+f(0) = \limsup_{\lambda \to 0^+} \frac{\alpha g(\lambda)}{1-\alpha \lambda \cdot g(\lambda)} \le \alpha. 
\end{align*}
By monotonicity, $g$ can be extended to $[0,1]$ with $g(0):=\lim_{\lambda \to 0^+} g(\lambda)  = \limsup_{\lambda \to 0^+} g(\lambda) \le 1$. This proves $g(\lambda) \le 1$ for all $\lambda \in [0,1]$, thus proving $f(\lambda) \le \frac{1}{1-\alpha \lambda}$. We conclude the proof by plugging back $\alpha = 1-\delta$. 
\end{proof}


\begin{proof}[Proof of~\Cref{prop:ode-boundary}]
    It suffices to show that for all non-maximal $S \in \XX$ and sufficiently small $\lambda > 0$, it holds
    \begin{align*}
        \mathrm{Cov}(\lambda*\mu^S) \preceq (1+(1-\delta) \lambda + o(\lambda)) \cdot \diag\tp{\*m(\lambda * \mu^S)}.
    \end{align*}
    By a straightforward calculation, the matrix $\mathrm{Cov}(\lambda * \mu^S) - \diag\tp{\*m(\lambda * \mu^S)}$ satisfies, for $i,j \in V_S$,
    \begin{align*}
        \tp{\mathrm{Cov}(\lambda*\mu^S) - \diag\tp{\*m(\lambda * \mu^S)}}(i,j) 
        &=  \begin{cases}
            \tp{\frac{\mu(S \cup \{i,j\})}{\mu(S)} - \frac{\mu(S \cup \{i\}) \mu(S \cup \{j\})}{\mu(S)^2}} \cdot \lambda^2 + o(\lambda^2), & i \neq j\\
            -\tp{\frac{\mu(S \cup \{i\})}{\mu(S)}}^2 \cdot \lambda^2 + o(\lambda^2), & i = j
        \end{cases}\\
        &= \tp{\diag(\*r_S) (M_S - I) \diag(\*r_S)}(i,j) \cdot \lambda^2 + o(\lambda^2).
    \end{align*}
    Similarly, for every $i$, it holds that
    \begin{align*}
      \*m(\lambda * \mu^S)_i = \frac{\mu(S\cup \set{i})}{\mu(S)} \lambda + o(\lambda) = \lambda \cdot \*r_S(i) + o(\lambda).
    \end{align*}
    Therefore, we have
    \begin{align}\label{eq:small-lambda}
        \nonumber \mathrm{Cov}(\lambda*\mu^S) - \diag\tp{\*m(\lambda * \mu^S)} &= \diag(\*r_S) (M_S-I) \diag(\*r_S) \cdot \lambda^2 + o(\lambda^2)\\
        \nonumber &\preceq (1-\delta) \diag(\*r_S) \cdot \lambda^2 + o(\lambda^2)\\
        &= (1 - \delta) \lambda \cdot \diag\tp{\*m(\lambda * \mu^S)} + o(\lambda^2).
    \end{align}
    We finish the proof of~\Cref{prop:ode-boundary} by noticing that $o(\lambda^2) = o(\lambda) \cdot \diag\tp{\*m(\lambda * \mu^S)}$.
\end{proof}

\subsection{Comparison between field dynamics and Glauber dynamics}
\label{subsec:comparison}

In this section, we prove~\Cref{lem:limit-field}.
The proof relies on the observation that the field dynamics $P^{\mathrm{FD}}_{1 \leftrightarrow 1-\epsilon}$, when run every $\epsilon$ time, converges roughly to the continuous-time Glauber dynamics $P$ as $\epsilon$ approaches zero.
More precisely, if we run the field dynamics $P^{\mathrm{FD}}_{1 \leftrightarrow 1-\epsilon}$ in continuous time and wait for an exponentially distributed random amount of time with mean $\eps$ between two adjacent runs, then the process will converge to some variant of the continuous-time Glauber dynamics.
In other words, one may view the field dynamics as a discretization of the continuous-time Glauber dynamics.
Therefore, we can naturally relate key quantities associated with the field dynamics to those for the Glauber dynamics, especially the generator and the Dirichlet form, allowing us to establish the comparison result \Cref{lem:limit-field}.

\begin{proof}[Proof of~\Cref{lem:limit-field}]
Note that the transition probability of field dynamics satisfies
\begin{align}\label{eq:transition-prob}
    \nonumber\forall S\neq T \in \XX, \quad (P^{\mathrm{FD}}_{1 \leftrightarrow 1-\epsilon} - I)(S,T) &= 
    \begin{cases}
        \epsilon + o(\epsilon) & S = T \cup \{u\} \text{ for some $u \in [n] \setminus T$}\\  
        \frac{\mu(T)}{\mu(S)} \cdot \epsilon + o(\epsilon) &  T = S \cup \{u\} \text{ for some $u \in [n] \setminus S$}\\
        o(\epsilon) & \text{otherwise}
    \end{cases}\\
    \nonumber &= \frac{\mu(S)+\mu(T)}{\mu(S \cap T)} n \epsilon \cdot (P-I)(S,T) + o(\epsilon)\\
    &\le (1+r_{\max}) n \epsilon \cdot (P-I)(S,T) + o(\epsilon).
\end{align}
Hence, for any $f \in \mathbb{R}^\Omega$, the Dirichlet form of the Glauber dynamics $P$ on $\mu$ satisfies
\begin{align*}
\+E_P(f,f) &= \inner{f}{(I-P)f}_{\mu} =
\frac{1}{2} \sum_{S,T \in \Omega} \mu(S) P(S,T) \tp{f(S) - f(T)}^2 \\
&= \frac{1}{2} \sum_{S,T \in \Omega} \mu(S) (P-I)(S,T) \tp{f(S) - f(T)}^2\\
&\overset{\eqref{eq:transition-prob}}{\ge} \frac{1}{2(1+r_{\max})n\epsilon} \sum_{S,T \in \Omega} \mu(S) (P^{\mathrm{FD}}_{1 \leftrightarrow 1-\epsilon} - I)(S,T) (f(S)-f(T))^2 + o(1)\\
&= \frac{1}{(1+r_{\max}) n\epsilon} \+E_{P^{\mathrm{FD}}_{1 \leftrightarrow 1-\epsilon}}(f,f) + o(1) \ge \frac{\delta(\epsilon)}{(1+r_{\max}) n \epsilon} \cdot \Var[\mu]{f} + o(1).
\end{align*}
By setting $\epsilon \to 0^+$, we have
\begin{align*}
    \forall f \in \mathbb{R}^\Omega,\quad \+E_P(f,f) \ge \frac{\delta^\star}{(1+r_{\max})n} \cdot \Var[\mu]{f}.
\end{align*}
This proves the Poincar\'e inequality with constant $\frac{\delta^\star}{(1+r_{\max})n}$. The proof for the modified log-Sobolev inequality follows from an identical argument.
\end{proof}

\section{Applications}

\subsection{Random cluster model}
\label{subsec:RC}
Let $\+M=([n],\+I)$ be a matroid on ground set $[n]$ with the collection of independent sets $\+I = \+I(M)$. The rank function $\mathrm{rk}:2^{[n]} \to \mathbb{Z}$ takes a subset $S \subseteq [n]$ as input, and output the maximum size of independent set $I \in \+I$ contained in $S$. Fix a constant $q > 0$ and weight $\lambda > 0$. The distribution $\mu_{\+M,q,\lambda}$ of the random cluster model is given by
\begin{align*}
\forall S \subseteq [n], \quad \mu_{\+M,q,\lambda}(S) \propto q^{- \mathrm{rk}(S)} \lambda^{\abs{S}}. 
\end{align*}
By taking $\lambda^\star = q \lambda$ and letting $q \to 0^+$, the distribution $\mu_{\+M,q,\lambda^\star}$ converges to the distribution $\mu_{\+I,\lambda}$ of $\lambda$-weighted independent sets. Specifically, the distribution $\mu_{\+I,\lambda}$ is supported over $\+I$, and satisfies
\begin{align*}
\forall I \in \+I, \quad \mu_{\+I,\lambda}(I) \propto {\lambda}^{\abs{I}}.
\end{align*}

In recent developments of the polynomial paradigm~\cite{branden2020lorentzian,anari2019logconcaveII}, the authors established the log-concavity for the homogenized generating polynomial of random cluster model when $q \le 1$. This polynomial is defined as
\begin{align*}
    Z_{\+M}(z_0,z_1,\ldots,z_n) = \sum_{A \subseteq [n]} q^{-\mathrm{rk}(A)} \cdot z_0^{n-\abs{A}} \cdot \prod_{i \in A} z_i.
\end{align*}
As a result, an optimal $O(n \log n)$ mixing time bound for the \emph{polarized walk} (i.e. the down-up walk for the polarized polynomial \cite{CGZZ24,anari2024trickle}) on the random cluster model has been derived when $q \le 1$~\cite{anari2019logconcaveII, cryan2019modified}. However, the tight mixing time bound for the Glauber dynamics remains unclear. There exist two partial responses for this question: Mousa showed an $O_{\lambda,q}(n^2 \log n)$ mixing time bound~\cite{mousa2022local}, and Liu derived an $O(n^{1+1/\sqrt{q}} \log n)$ mixing time bound~\cite{liu2023spectral}. 

In the following theorem, we address this problem by showing an $O_{q,\lambda}(n \log n)$ mixing time bound. As a side product, we also prove the log-concavity of the generating polynomial of $\mu_{\+M,q,\lambda}$ by \cref{prop:log-concave}, which is inhomogeneous.

\begin{theorem}\label{thm:random-cluster}
    Let $\+M=([n],\+I)$ be a matroid with  the collection of independent sets $\+I$.
    \begin{enumerate}

      \item  Let $\mu_{\+M,q,\lambda}$ be the distribution of random cluster model with $q \in (0,1]$ and $\lambda > 0$.
       The Glauber dynamics on $\mu_{\+M,q,\lambda}$ satisfies the modified log-Sobolev inequality with constant 
       \begin{align*}
           \frac{1}{n} \cdot \max \left \{\frac{1}{1+q^{-1}\lambda} , \frac{1}{1+\lambda^{-1}} \right \},
       \end{align*}
       implying an $O_{q,\lambda}(n \log n )$ mixing time.
       
       \item Let $\mu_{\+I,\lambda}$ be the distribution of $\lambda$-weighted independent sets. The Glauber dynamics on $\mu_{\+I,\lambda}$ satisfies the modified log-Sobolev inequality with constant $\frac{1}{n(1+\lambda)}$, implying an $O_\lambda(n (\log r + \log \log n))$
       mixing time, where $r$ is the rank of matroid $\+M$.
    \end{enumerate}
\end{theorem}

It is interesting to see if one could improve the mixing time bound $O(n (\log r + \log \log n))$ for weighted independent sets to $O(n \log r)$, matching the result in \cite{anari2021logconcaveIV} for sampling bases of a matroid. Here the dependency on $\log \log n$ comes from $\log \log (1/\mu_{\min})$.

\begin{proof}[Proof of \cref{thm:random-cluster}]
Fix any $S \subsetneq [n]$ and let $\+M / S$ denote the matroid contraction. It holds that
\begin{align*}
    \forall i,j \in V_S,\quad M_S(i,j) = 
    \begin{cases}
    0, & i = j;\\
    q-1, & \{i\},\{j\} \in \+I(\+M/S) \text{ and } \{i,j\} \not\in \+I(M/S);\\
    0, & \text{otherwise.}
    \end{cases}
\end{align*}
According to standard matroid theory, $\set{i,j}\not\in \+I(M/S)$ gives an equivalence relation $\simeq$ between singletons in $\+I(M/S)$. 
Therefore, $M_S$ can be represented as
    \begin{align*}
        M_S = I + (q-1)
        \begin{pmatrix}
            J_1 & & & \\
            & J_2 & & \\
            & & \ddots & \\
            & & & J_k
        \end{pmatrix} \preceq I,
    \end{align*}
    where the matrices $J_1, \dots, J_k$ are all-one matrices of potentially distinct sizes, corresponding to equivalence classes defined by the equivalence relation $\simeq$. 
    By~\Cref{thm:main-0} and~\Cref{lem:limit-field}, the Glauber dynamics $P$ satisfies the modified log-Sobolev inequality with constant
    \begin{align*}
        \rho_0(P) \ge \frac{1}{(1+r_{\max}) n} = \frac{1}{(1+q^{-1}\lambda)n}. 
    \end{align*}
    Note that the distribution $\mu_{\+M,q,\lambda}$ random cluster model can be treated as a distribution of random cluster model on its dual matroid $\overline{\+M}$ with a different choice of parameters.
    In particular, $\ol{\+M}$ is another matroid on the same ground set $[n]$ with the rank function $\ol{\-{rk}}: 2^{[n]} \to \=N$ satisfying
    \begin{align*}
      \forall S \subseteq [n], \quad \ol{\-{rk}}(S) = \abs{S} + \-{rk}([n]\setminus S) - \-{rk}([n]).
    \end{align*}
    Hence, let $S \sim \mu_{\+M,q,\lambda}$, the sample $[n]\setminus S$ follows the law of $\mu_{\overline{\+M},q,q/\lambda}$, which means $\mu_{\+M,q,\lambda}$ and  $\mu_{\overline{\+M},q,q/\lambda}$ are essentially the same distribution.
    This implies that the Glauber dynamics $P$ also satisfies the modified log-Sobolev inequality with constant
    \begin{align*}
        \rho_0(P) \ge \frac{1}{(1+r_{\max}) n} = \frac{1}{(1+\lambda^{-1})n}. 
    \end{align*}
    This proves the first part of~\Cref{thm:random-cluster}. The modified log-Sobolev constant and mixing time upper bound of Glauber dynamics for the distribution $\mu_{\+I}$ follows from the fact that $\mu_{\+M,q,q\lambda}$ converges to $\mu_{\+I}$ when $q$ approaches zero.
\end{proof}

\subsection{Antiferromagnetic two-spin system}
\label{subsec:anti-2-spin}
Given a graph $G = (V,E)$ of $n$ vertices, the two-spin system on $G$ is defined as a distribution $\mu_G$ over $2^n$ configurations, which are assignments of spins $0,1$ to the vertices.
Formally, it is defined by two edge parameters $\beta\geq 0, \gamma > 0$ and one vertex parameter $\lambda > 0$. The parameter $\beta$ corresponds to the strength of interaction between neighboring $1$-spins, $\gamma$ corresponds to the strength of interaction between neighboring $0$-spins, and $\lambda$ is the external field assigned to each $1$-spin.
The weight of each configuration $\sigma \in \{0,1\}^V$ is assigned as follows:
$$
    w_{G}(\sigma) := \beta^{m_1(\sigma)}\gamma^{m_0(\sigma)}\lambda^{n_1(\sigma)},
$$
where $m_s(\sigma) = |\{(u,v)\in E\mid \sigma(u) = \sigma(v) = s\}|$ and $n_s(\sigma) = |\{v\in V\mid \sigma(v) = s\}|$ for $s=0,1$.
The Gibbs distribution $\mu_G$ and the partition function $Z_G$ is defined as
$$
    \mu_G(\sigma) := \frac{w_G(\sigma)}{Z_G},\quad \text{where } Z_G := \sum_{\sigma\in \{0,1\}^V} w_G(\sigma).
$$
We may treat $\mu_G$ as a distribution over subsets of $V$, and $\gamma > 0$ implies that $\mu_G$ is indeed a distribution over a downward closed set family.
When $\beta\gamma > 1$, the model is called \emph{ferromagnetic} and it is \emph{antiferromagnetic} if $\beta\gamma < 1$.
We mention two most well-studied examples:
When $\beta = 0, \gamma = 1$, we obtain the hardcore model with fugacity $\lambda$; when $\beta = \gamma < 1$, we obtain the antiferromagnetic Ising model with inverse temperature $\beta < 1$ and external field $\lambda$.

In the last decades, a fascinating computational phase transition phenomenon was established regarding sampling from antiferromagnetic two-spin systems. 
Take the hardcore model as an example. The critical threshold $\lambda_c(\Delta) \approx \frac{\e}{\Delta}$ characterizes the uniqueness/non-uniqueness of the Gibbs measure on the infinite $\Delta$-regular tree. 
It was shown that, for graphs of maximum degree $\Delta$, a polynomial-time sampler exists when $\lambda < \lambda_c(\Delta)$~\cite{weitz2006counting}, and there is no polynomial-time sampler when $\lambda > \lambda_c(\Delta)$ assuming $\textsf{NP}\neq \textsf{RP}$~\cite{sly2010computational}.
In a recent line of studies, it was further established that the Glauber dynamics exhibits an optimal $O(n \log n)$ mixing time when $\lambda$ is within the uniqueness regime, i.e., $\lambda < \lambda_c(\Delta)$~\cite{ALO2021spectral,chen2021optimal,chen2022localization,CFYZ22optimal}. 
Similar results were also derived for general antiferromagnetic two-spin systems~\cite{LLY13,sly2014counting,galanis2016inapproximability,chen2021optimal,CFYZ22optimal}. 
    
Despite the worst-case hardness results, Glauber dynamics is observed to perform well on random instances beyond the uniqueness threshold.
For example, Glauber dynamics for the antiferromagnetic Ising model on random $\Delta$-regular graphs mixes optimally with high probability when the graph is sampled uniformly at random and the inverse temperature satisfies $1 - \beta \lesssim \frac{1}{2\sqrt{\Delta}} $, much larger than the uniqueness regime $1 - \beta \le \frac{2}{\Delta}$~\cite{anari2024trickle}. 

In the following theorem, we establish a new criterion related to the minimum eigenvalue of the underlying graph for rapid mixing of Glauber dynamics on general antiferromagnetic two-spin systems. 
For the hardcore model this gives \cref{thm:HC-main-lambda-min} from the introduction. 
As a corollary, we prove the rapid mixing of Glauber dynamics for the hardcore model and the Ising model on random regular graphs beyond the uniqueness threshold.

\begin{theorem}\label{thm:mixing_antifero_2spin}
Let $\mu$ be the Gibbs distribution of an antiferromagnetic two-spin system specified by parameters $(\beta,\gamma,\lambda)$ on a $\Delta$-regular graph $G=(V,E)$ of $n$ vertices.
Furthermore, let $-\lambda^\star$ be the minimum eigenvalue of $A_G$, the adjacency matrix of $G$.
\begin{enumerate}  
    \item If $1 - \beta\gamma \leq \frac 1{\lambda^\star}$, then the modified log-Sobolev inequality for Glauber dynamics $P_{\mathrm{GD}}$ holds with constant
    \begin{align*}
    \rho_0(P_{\mathrm{GD}}) \ge \frac{1}{(1+\lambda/\gamma^\Delta) n},
    \end{align*}
    implying an $O(n \log n)$ mixing time.

    \item If $1 - \beta\gamma > \frac 1{\lambda^\star}$ and $\lambda \leq \frac{(1-\delta)\gamma^\Delta}{(1-\beta\gamma)\lambda^\star - 1}$ for some $\delta \in (0,1)$, then the Poincar\'e inequality for Glauber dynamics $P_{\mathrm{GD}}$ holds with constant
    \begin{align*}
    \lambda(P_{\mathrm{GD}}) \ge \frac{\delta}{(1+\lambda/\gamma^\Delta) n},
    \end{align*}
    implying an $O(n^2/\delta)$ mixing time.
\end{enumerate}
\end{theorem}

    
\begin{proof}
    Fix non-maximal $S \in \XX$, where $\XX$ is the support of $\mu$. It holds that
    \begin{align*}
    \*r_S(i) = \frac{\mu^S(\{i\})}{\mu^S(\emptyset)} = \frac{\lambda (\beta\gamma)^{d_i^S}}{\gamma^\Delta} \quad \text{and}\quad 
    M_S(i,j) =  (\beta\gamma)^{\ind \set{ \{i,j\} \in E}} - 1,
    \end{align*}
    where $d_i^S$ is the number of vertices in $S$ that is adjacent to $i$.
    Let $A_{G,S}$ be the adjacency matrix of subgraph $G[V_S]$ where $V_S = \{u\notin S\mid \{u\}\cup S \in \XX\}$.
    By the Cauchy interlacing theorem, it holds 
    $$-\lambda_{\min}(A_{G,S}) \leq -\lambda_{\min}(A_{G,\emptyset})  = \lambda^\star.$$
    Thus, the maximum eigenvalue of $M_S$ satisfies
    \begin{align}\label{eq:max-eigenvalue}
        \lambda_{\max}\tp{M_S} = \lambda_{\max}\tp{(\beta \gamma -1) A_{G,S}} \le (1-\beta \gamma) \lambda^\star.
    \end{align}
    Meanwhile, the minimum eigenvalue of $\diag(\*r_S)^{-1}$ is lower bounded as follows:
    \begin{align}\label{eq:min-eigenvalue}
        \lambda_{\min}\tp{\diag\tp{\*r_S}^{-1}} \ge \min_{0 \le \ell \le \Delta} \frac{\gamma^\Delta}{\lambda (\beta \gamma)^\ell} = \frac{\gamma^\Delta}{\lambda}.
    \end{align}
    Combining~\eqref{eq:max-eigenvalue} and~\eqref{eq:min-eigenvalue}, we have
    \begin{align*}
        M_S - I \preceq 
        \begin{cases}
            0, & \text{if } (1-\beta\gamma) \lambda^\star \le 1; \\
            \left( (1-\beta\gamma) \lambda^\star - 1 \right) \frac{\lambda}{\gamma^\Delta} \cdot \diag\tp{\*r_S}^{-1}, & \text{if } (1-\beta\gamma) \lambda^\star > 1.
        \end{cases}
    \end{align*}
    Together with~\Cref{thm:main-0,thm:main-delta}, this completes the proof of~\Cref{thm:mixing_antifero_2spin}.
\end{proof}

Our main applications lie in sampling from hardcore model (i.e., $\lambda$-weighted independent sets) and antiferromagnetic Ising model on random $\Delta$-regular graphs. 
To apply~\Cref{thm:mixing_antifero_2spin}, the following result on the spectrum of adjacency matrices of random $\Delta$-regular graphs is required.
\begin{lemma}[\cite{friedman2008proof,bordenave2020new}]\label{lem:lammin_dregular}
    For any constant $\Delta\geq 3$, with probability $1 - o_n(1)$ over the choice of a random $\Delta$-regular graph $G = (V,E)$ on $n$ vertices, we have that
    $$
        \abs{\lambda_{\min}(A_G)} \leq (2 + o_n(1))\sqrt{\Delta-1}.
    $$
\end{lemma}

The mixing time results for hardcore model and antiferromagnetic Ising model are thereby straightforward corollaries of~\Cref{thm:mixing_antifero_2spin,lem:lammin_dregular}.
\begin{corollary}[Hardcore model]\label{cor:SI_hardcore}
    Let $\Delta \ge 3$ and $\delta \ge 0$ be fixed constants, and $\mu$ be the Gibbs distribution of hardcore model with fugacity $\lambda \le \frac{1-\delta}{2\sqrt{\Delta-1}-1}$ on a random $\Delta$-regular graph $G$. With probability $1 - o_n(1)$ over the choice of $G$, the mixing time of Glauber dynamics for $\mu$ satisfies 
    \begin{align*}
    T_{\mathrm{mix}} = O_{\Delta,\lambda} \tp{ \frac{n^2}{\delta} }.
    \end{align*}
\end{corollary}

\begin{corollary}[Ising model]\label{cor:SI_Ising}
    Let $\Delta \ge 3$ be a fixed constant, and $\mu$ be the Gibbs distribution of antiferromagnetic Ising model specified by $\beta \ge 1 - \frac{1}{4\sqrt{\Delta-1}}$ and arbitrary external field $\lambda > 0$ on a random $\Delta$-regular graph $G$. With probability $1 - o_n(1)$ over the choice of $G$, the mixing time of Glauber dynamics for $\mu$ satisfies 
    \begin{align*}
    T_{\mathrm{mix}} = O_{\Delta,\beta,\lambda}\tp{n \log n}.
    \end{align*}
\end{corollary}

\begin{remark}
    Our condition $1-\beta \le \frac{1}{4\sqrt{\Delta-1}}$ in \Cref{cor:SI_Ising} implies that $1-\beta^2 \le \frac{1-o_n(1)}{2\sqrt{\Delta-1}}$, which is required to apply \cref{thm:mixing_antifero_2spin}.
    A more precise mixing time upper bound for the Ising model without external fields (i.e., $\lambda=1$), following from \Cref{thm:mixing_antifero_2spin}, is given by 
    \begin{align*}
        T_{\mathrm{mix}} = O\tp{ \frac{n}{\beta^\Delta} \log n} = O\tp{ e^{\sqrt{\Delta}} n \log n}.
    \end{align*}
     In~\cite{anari2024trickle}, the authors established the modified log-Sobolev inequalities with constant $\exp(-\e^{\sqrt{\Delta}})$ with inverse temperature $1 - \beta \lesssim \frac{1}{2\sqrt{\Delta-1}}$. 
     \Cref{cor:SI_Ising} matches the order of the regime in~\cite{anari2024trickle}, and significantly improves the constant from double exponential to single exponential.
\end{remark}

\subsection{Binary symmetric Holant problem}
\label{subsec:Holant}
The binary Holant problem is a graphical model defined over the subsets of edges of a given graph. 
Many classical problems can be defined within the family of Holant problems, including ($b$-)matchings, ($b$-)edge covers, high-temperature expansions of Ising model, etc. 
A long line of research was dedicated to interesting Holant problems spanning several decades~\cite{jerrum1993ising,jerrum2003counting,cai2009holant,guo2013complexity,huang2016canonical,guo2021zeros,feng2023swendsen,chen2023nearlinear,chen2024fast}. Notable examples include the rapid mixing of Sinclair--Jerrum chain in the monomer-dimer model on arbitrary graphs~\cite{jerrum2003counting}, and the optimal mixing of the Glauber dynamics for the distributions of weighted $b$-matchings and $b$-edge covers on bounded degree graphs~\cite{chen2024fast}.
In this section, we consider the binary symmetric Holant problem and revisit some classical examples.

Given a graph $G = (V,E)$ of $n$ vertices, let $d_v$ denote the degree of $v$ and let $E_v$ denote $\{e\in E\mid v \in e\}$. The binary symmetric Holant problem is specified by external fields $\*\lambda = (\lambda_e)_{e\in E}$ on edges, and a family of constraint functions $\*f = (f_v)_{v\in V}$ on vertices. Each of these constraints function maps an integer to a non-negative real number.
The sequence $f_v = [f_v(0), f_v(1),...,f_v(d_v)]$ is called the \emph{signature} of $v$, where $d_v$ is the degree of vertex $v \in V$.
The Gibbs distribution $\mu_{G,\*f,\*\lambda}$ and the partition function $Z_{G,\*f,\*\lambda}$ is defined as follows:
\begin{align*}
    &\forall S\subseteq E, \quad \mu_{G,\*f,\*\lambda}(S) = \frac 1{Z_{G,\*f,\*\lambda}}\prod_{v\in V}f_v(|E_v \cap S|)\prod_{e\in S}\lambda_e\\
    &\text{where} \quad Z_{G,\*f,\*\lambda} = \sum_{S\subseteq E}\prod_{v\in V}f_v(|E_v \cap S|)\prod_{e\in S}\lambda_e.
\end{align*}
In this work, we consider uniform external fields $\*\lambda = \lambda \*1$ for simplicity, where $\lambda > 0$ is applied to all edges.
If $f_v(i) = \ind[i\leq 1]$ at every vertex $v$, then we recover the distribution of $\lambda$-weighted matchings of $G$, i.e., the monomer-dimer model. Similarly, if $f_v(i) = \ind[i\leq b]$ for all $v$, we obtain the distribution of $\lambda$-weighted $b$-matchings.

Our main application is to sample from binary symmetric Holant problems with log-concave signatures, which are defined as follows.
\begin{definition}[\text{\cite[Definition 2]{chen2024fast}}]\label{def:log_concave_signature}
    A sequence $f = [f(0),f(1),...,f(d)]$ of non-negative real numbers is called a \emph{downward closed log-concave} signature if it satisfies the following conditions:
    \begin{enumerate}
        \item Log-concavity: $f(k)^2\geq f(k-1)f(k+1)$ for all $1\leq k\leq d-1$;
        \item No internal zeros: if $f(k_1) > 0$ and $f(k_2) > 0$ for some $0\leq k_1 < k_2\leq d$, then $f(k) > 0$ for all $k_1\leq k \leq k_2$;
        \item Downward closed: $f(0) > 0$.
    \end{enumerate}
\end{definition}

For instance, the signature of $b$-matchings $f(i) = \ind[i\leq b]$ is log-concave.
We remark that the third condition ensures $\mu_{G,\*f,\*\lambda}$ is a distribution supported on downward closed set families, which was also required in \cite{chen2024fast}.

For Holant problems with log-concave signatures, we prove the following mixing results via \Cref{thm:main-0,thm:main-delta}.

\begin{theorem}\label{thm:Mixing_holant}
    Consider a binary symmetric Holant problem defined by a simple graph $G = (V,E)$ of $n$ vertices and $m$ edges, downward closed log-concave signatures $\*f$, and uniform external fields $\*\lambda = \lambda \*1$ with $\lambda > 0$. Furthermore, define $Q$ and $R$ as follows:
    $$
        Q = \min_{\substack{u\in V\\0\leq k\leq d_u-2}}\frac{f_u(k)f_u(k+2)}{f_u(k+1)^2}
        \quad\text{and}\quad
        R = \max_{u\in V} \frac{f_u(1)}{f_u(0)}.
    $$
    \begin{enumerate}
        \item If $Q \geq \frac 12$, then the modified log-Sobolev inequality for Glauber dynamics $P_{\mathrm{GD}}$ holds with constant 
        \begin{align*}
        \rho_0(P_{\mathrm{GD}}) \ge \frac{1}{(1+R)m},
        \end{align*}
        implying an $O(m \log m)$ mixing time.
        
        \item If $Q < \frac 12$ and $\lambda \leq \frac{1-\delta}{(1-2Q)R^2}$ for some $\delta \in (0,1)$, then the Poincar\'e inequality for Glauber dynamics $P_{\mathrm{GD}}$ holds with constant
        \begin{align*}
            \lambda(P_{\mathrm{GD}}) \ge \frac{\delta}{(1+R) m},
        \end{align*}
        implying an $O(m^2/\delta)$ mixing time.
    \end{enumerate}
\end{theorem}

The following technical lemma is helpful to us in the proof of \cref{thm:Mixing_holant}.
\begin{lemma}\label{lem:line_graph_adj_bound}
    Given a simple graph $G = (V,E)$ and a positive vector $\*w \in \R_{\geq 0}^V$, let $\*M\in \R^{E\times E}$ denote the following matrix:
    $$
        \*M(e_1,e_2) = \begin{cases}
            \*w(u), & \text{if } e_1\cap e_2 = \{u\};
            \\0, &\mathrm{otherwise};
        \end{cases}
    $$
    and let $\*D = \diag\tp{\*w(u) + \*w(v)}_{\{u,v\}\in E}$.
    Then we have that
    $\*M \succeq -\*D$.
\end{lemma}
\begin{proof}
    Let $\*B \in \R^{E\times V}$ denote the matrix which satisfies that
    $\*B(e,u) = \ind[u\in e] \sqrt{\*w(u)}$ for any $e\in E,u\in V$.
    Then the following always holds
    $$
        \*M = \*B\*B^\intercal - \*D \succeq -\*D,
    $$
    since $\*B\*B^\intercal$ is positive semidefinite.
\end{proof}
Now we use~\Cref{lem:line_graph_adj_bound} to prove the mixing time of Glauber dynamics for binary symmetric Holant problems.

\begin{proof}[Proof of \cref{thm:Mixing_holant}]
    Fix non-maximal $S \in \XX$, where $\XX$ is the support of the Gibbs distribution of the binary symmetric Holant problem.
    Let $E_S := \{e \in E\setminus S: S\cup \{e\} \in \+X\}$, $d_u$ denote the degree of vertex $u \in V$, and $k_v = k_v^S$ denote the number of edges in $S$ that are incident to $v$.
    It holds that:
    \begin{align*}
        \forall e = (u,v) \in E_S,\quad  \*r_S(e) &=  \frac{\lambda f_u(k_u+1)f_v(k_v+1)}{f_u(k_u)f_v(k_v)}; \\
       \forall e_i,e_j\in E_S,\quad  M_S(e_i,e_j) &= \begin{cases}
        \frac{f_u(k_u)f_u(k_u+2)}{f_u^2(k_u+1)} - 1, &  \text{if } e_i\cap e_j = \{u\};
        \\ 0, & \text{otherwise.}
       \end{cases}
    \end{align*}
    For any $u \in V$, let $g_S(u) = \frac{f_u(k_u)f_u(k_u+2)}{f_u^2(k_u+1)}$ if $k_u+2 \le d_u$, and $g_S(u) = 0$ otherwise.
    Log-concavity of $\*f$ implies that $g_S(u) \leq 1$ for any $u\in V$.
    By~\Cref{lem:line_graph_adj_bound}, we have
    $$
        M_S \preceq D_S + I
    $$
    where $D_S\in \R^{E_S\times E_S}$ denotes $\diag\{1-g_S(u)-g_S(v)\}_{e = \{u,v\} \in E_S}$.
    Therefore, it holds
    \begin{align*}
        \lambda_{\max}(M_S - I) \le \lambda_{\max}\tp{D_S} \le 1-2Q.
    \end{align*}
    Meanwhile, the minimum eigenvalue of $\diag\tp{\*r_S}^{-1}$ is lower bounded as follows:
    \begin{align*}
        \lambda_{\min}(\diag\tp{\*r_S}^{-1}) \ge \min_{u,v \in V} \frac{f_u(k_u) f_v(k_v)}{\lambda f_u(k_u+1) f_v(k_v+1)} \ge \frac{1}{\lambda R^2},
    \end{align*}
    where the second inequality follows from the log-concavity of signature $\*f$.
    Therefore,
    \begin{align}\label{eq:M_S-I}
        M_S - I \preceq 
        \begin{cases}
            0, & \text{if } Q \ge \frac{1}{2};\\
            (1-2Q)\lambda R^2 \cdot \diag\tp{\*r_S}^{-1}, & \text{if } Q < \frac{1}{2}. 
        \end{cases}
    \end{align}
    Plugging in~\Cref{thm:main-0,thm:main-delta}, we conclude the proof.
\end{proof}
As immediate consequences of \cref{thm:Mixing_holant}, we obtain new rapid mixing results for Glauber dynamics on the monomer-dimer model and the $b$-matching problem, thus proving \cref{thm:matching-main} from the introduction.

\begin{corollary}[matching]\label{cor:matching}
Let $\mu$ be the Gibbs distribution of the monomer-dimer model specified by a \emph{simple} graph $G=(V,E)$ and fugacity $\lambda \le 1- \delta$ where $\delta \in [0,1)$. 
Then the mixing time of the Glauber dynamics on $\mu$ satisfies
\begin{align*}
    T_{\mathrm{mix}} = O_\lambda \tp{ \min\left\{ \frac{1}{\delta}, \sqrt{\Delta} \right\} mn \log n},
\end{align*}
where $m = |E|$ is the number of edges in graph $G$ and $\Delta$ is the maximum degree of $G$.
\end{corollary}

\begin{proof}
The first part ($1/\delta$ bound) follows from~\Cref{thm:Mixing_holant} with $Q = 0$ and $R = 1$, and $\mu_{\min} \ge \frac{\lambda^n}{m^n}$ when $\lambda \le 1$. 
Meanwhile, when $\lambda \le 1$, by~\eqref{eq:M_S-I},~\Cref{thm:cor-bound} and the fact that the distribution $\mu$ is $(2\sqrt{1+\lambda\Delta})$-spectrally independent~\cite{chen2021optimal} (see also \cref{rmk:SI-SS} for the comparison between spectral stability with respect to field dynamics and spectral independence), it holds
\begin{align*}
    \Cov\left( (1-\theta)*\mu^S \right) \preceq \underbrace{\min\set{\frac{1}{\theta},2\sqrt{1+\lambda\Delta}}}_{=: C(\theta)} \cdot \diag \tp{ \*m\tp{ (1-\theta) * \mu^S } }.
\end{align*}
Together with~\Cref{thm:mixing-bound}, the Poincar\'e inequality holds with constant
\begin{align*}
    \lambda(P_{1 \leftrightarrow \theta}^{\mathrm{FD}}) \ge \exp\tp{-\int_0^{\theta} \frac{C(\eta)}{1-\eta} \dif \eta}  
    = \Omega\tp{\frac{1-\theta}{\sqrt{\Delta}}}.
\end{align*}
By~\Cref{lem:limit-field}, this implies a Poincar\'e inequality for the Glauber dynamics with constant 
\begin{align}\label{eq:GD-poincare-2}
\lambda(P_\mathsc{gd}) = \Omega \left( \frac{1}{\sqrt{\Delta} m} \right).
\end{align}
This proves the desired lower bound on the spectral gap of the Glauber dynamics when $\lambda \le 1$.
\end{proof}

\begin{corollary}[$b$-matchings]\label{cor:SI_b_matchings}
    Let $\mu$ be the distribution of $\lambda$-weighted $b$-matchings of a \emph{simple} graph $G=(V,E)$ with $\lambda \le 1 -\delta$ where $\delta \in [0,1)$.
    Then the mixing time of the Glauber dynamics on $\mu$ satisfies
    \begin{align*}
        T_{\mathrm{mix}} = O_\lambda \tp{ \min\left\{ \frac{1}{\delta}, \Delta^b, bn \right\} \cdot b m n \log n},
    \end{align*}
    where $m = |E|$ is the number of edges in graph $G$ and $\Delta$ is the maximum degree of $G$.
\end{corollary}
\Cref{cor:SI_b_matchings} follows similarly by $\mu_{\min} \ge \frac{\lambda^{bn}}{m^{bn}}$ and the fact that $\mu$ is $O(\min\{\Delta^b, bn\})$-spectrally independent~\cite{chen2024fast}.

\subsection{Determinantal point process}
\label{subsec:DPP}

The determinantal point process was first introduced in~\cite{benard1973detection} to describe ``fermions''  in thermal equilibrium. It has since found many applications in machine learning, as comprehensively surveyed in~\cite{kulesza2012determinantal}.
Let $L \in \mathbb{R}^{n \times n}$ be a symmetric positive semidefinite matrix. 
Given the parameter $\alpha \in [0,1]$, the determinantal point process defines a distribution $\mu$ supported over $\{0,1\}^{n}$ given as follows:
\begin{align*}
    \forall S \subseteq [n], \quad \mu(S) \propto \left( \mathrm{det}(L_{S,S}) \right)^\alpha,
\end{align*}
with the convention $\mathrm{det}(L_{\emptyset,\emptyset}) = 1$.
Here, $L_{S,S}$ is the submatrix of $L$ by taking rows and columns with indices in $S$. 
When $\alpha=1$, the distribution is easier to study since the partition function admits a simple determinantal expression and is in fact strongly log-concave as a polynomial in external fields.

Recent years have seen several advances in sampling from determinantal point processes (DPPs). 
For instances, in~\cite{AJKPV2022entropic}, the authors showed an $\widetilde{O}(n^2)$ mixing of the $2$-block dynamics for distribution $\mu$; In~\cite{anari2022optimal}, the authors provides an efficient sampler for $k$-DPP distribution. In this section, we prove an optimal mixing time for the Glauber dynamics (i.e., $1$-block dynamics).

\begin{theorem}
\label{thm:dpp-alpha}
    Let $L \in \mathbb{R}^{n \times n}$ be a symmetric positive semidefinite matrix and $\mu$ be the distribution of the determinantal point process with the parameter $\alpha \in [0,1]$. Then the generating polynomial $g_\mu$ is strongly log-concave, thus implying that $\mu$ satisfies the modified log-Sobolev inequality with constant at least $\frac{1}{(1+\max_{i} L_{i,i}^\alpha)n}$ and the mixing time of Glauber dynamics for $\mu$ is $O\tp{(1+\max_{i} L_{i,i}^\alpha)n \log n}$.
\end{theorem}

To prove \cref{thm:dpp-alpha}, we first consider the special case where $\alpha = 1$, i.e., the standard DPP. 
It is known that the generating polynomial $g_{\mu}$ for the standard DPP is real stable, thus strongly log-concave; such an implication can be established via a homogenization argument. In this section, we include an alternative proof by using~\Cref{prop:log-concave,thm:main-0}.

\begin{lemma}\label{thm:dpp_alpha=1}
Let $L \in \mathbb{R}^{n \times n}$ be a symmetric positive semidefinite matrix and $\mu$ be the distribution of the determinantal point process with $\alpha = 1$. Then the generating polynomial $g_\mu$ is strongly log-concave, thus implying the modified log-Sobolev inequality with constant at least $\frac{1}{(1+\max_{i} L_{i,i})n}$ and an $O\tp{(1+\max_{i} L_{i,i})n \log n}$ mixing time for Glauber dynamics.
\end{lemma}

\begin{proof}
   Fix a non-maximal set $S \subseteq [n]$ with $\mathrm{det}(L_{S,S}) > 0$. For any $i,j \in V_S$, we write the principal submatrix $L_{S \cup \{i,j\},S \cup \{i,j\} }$ as follows:
   \begin{align*}
   L_{S \cup \{i,j\}, S \cup \{i,j\} } = 
   \begin{bmatrix}
    L_S & u_i & u_j\\
    u_i^\intercal & \ell_{i,i} & \ell_{i,j}\\
    u_j^\intercal & \ell_{j,i} & \ell_{j,j}
   \end{bmatrix},
   \end{align*}
   where $u_i,u_j \in \mathbb{R}^{S}$. Therefore, the determinants of $L_{S \cup \{i\}, S \cup \{i\} }$ and $L_{S \cup \{i,j\}, S \cup \{i,j\}}$ satisfy:
   \begin{align*}
    \mathrm{det} \left( L_{S \cup \{i\}, S \cup \{i\} } \right) &= \mathrm{det}(L_{S,S}) \cdot \mathrm{det}\left( \ell_{i,i} - u_i^\intercal L_{S,S}^{-1} u_i \right) 
    = \mathrm{det}(L_{S,S}) \cdot \left( \ell_{i,i} - u_i^\intercal L_{S,S}^{-1} u_i \right); \\
       \mathrm{det} \left( L_{S \cup \{i,j\}, S \cup \{i,j\}} \right) &= \mathrm{det}(L_{S,S}) \cdot 
       \tp{ \left( \ell_{i,i} - u_i^\intercal L_{S,S}^{-1} u_i \right) \left( \ell_{j,j} - u_j^\intercal L_{S,S}^{-1} u_j \right) - \left( \ell_{i,j} - u_i^\intercal L_{S,S}^{-1} u_j \right)^2}.
   \end{align*}
   Hence, the matrix $M_S - I$ is given by
   \begin{align}\label{eq:M-I}
        \forall i,j \in V_S,\quad (M_S - I)(i,j) = -\frac{\left( \ell_{i,j} - u_i^\intercal L_{S,S}^{-1} u_j \right)^2}{\left( \ell_{i,i} - u_i^\intercal L_{S,S}^{-1} u_i \right) \left( \ell_{j,j} - u_j^\intercal L_{S,S}^{-1} u_j \right) }.
   \end{align}
   Note that the column vectors of $L_{S,V_S}$ are $u_i$ for $i \in V_S$. Therefore, we have
   \begin{align}\label{eq:N}
   N := \left[ \ell_{i,j} - u_i^\intercal L_{S,S}^{-1} u_j \right]_{i,j \in V_S} = L_{V_S, V_S} - L_{S, V_S}^\intercal L_{S,S}^{-1} L_{S, V_S} \succeq 0,
   \end{align}
   where the last inequality follows from that $L$ is positive semidefinite.
   Therefore, by~\eqref{eq:M-I} and~\eqref{eq:N}, the matrix $M_S - I$ satisfies:
   \begin{align*}
    M_S - I = - \Lambda^{-1} \cdot \tp{N \circ N} \cdot \Lambda^{-1} \preceq 0,
   \end{align*}
   where $\Lambda = \mathrm{diag}\tp{\ell_{i,i} - u_i^\intercal L_{S,S}^{-1} u_i}_{i \in V_S}$, $N \circ N$ denotes the Hadamard product of matrix $N$ with itself, and the last inequality follows from $N \succeq 0$ and the Schur product theorem.
   We then conclude the lemma from~\Cref{prop:log-concave}, \Cref{thm:main-0}, and that 
   \begin{align*}
       r_{\max} &= \max_{\substack{\text{non-maximal } S\\i \in V_S}} {\ell_{i,i} - u_i^\intercal L_{S,S}^{-1} u_i} \le \max_{i \in [n]}{\ell_{i,i}}. \qedhere
   \end{align*}
\end{proof}

Next, we prove \cref{lem:log-concave-g_alpha} which shows that strong log-concavity is preserved at higher temperatures.
Thus, we can generalize~\Cref{thm:dpp_alpha=1} to the determinantal point process with arbitrary $\alpha \in [0,1]$, establishing \cref{thm:dpp-alpha}.
We note that \cref{lem:log-concave-g_alpha} is analogous to \cite[Theorem 1.7]{anari2019logconcaveII} for homogeneous polynomials, and can be proved by homogenization arguments.
Here, and later for \cref{lem:log-concave-varphi}, we present direct proofs of both facts using Condition \eqref{eq:strong-assump-main} and \Cref{prop:log-concave}.

\begin{lemma}\label{lem:log-concave-g_alpha}
    Let $\XX \subseteq \{0,1\}^n$ be a non-empty downward closed set family, and let $\mu$ be a distribution fully supported on $\XX$.
    For $\alpha \in [0,1]$, let $\nu$ be a distribution on $\XX$ defined as
    \begin{align*}
        \nu(S) \propto \mu(S)^\alpha, \quad \forall S\in \+X.
    \end{align*}
    If the generating polynomial of $\mu$ is strongly log-concave, then the generating polynomial of $\nu$ is also strongly log-concave.
\end{lemma}
\begin{proof}
    By~\Cref{prop:log-concave}, it is sufficient to show that for any non-maximal $S\in \+X$,
    $$
        I - M_S(\mu) \succeq \*0 \implies I - M_S(\nu) \succeq \*0.
    $$
    Let $J = \*1\*1^\intercal$ be the all-$1$ matrix.
    It holds that
    \begin{align*}
      M_S(\nu) - I + J &= (M_S(\mu) - I + J)^{\circ \alpha},
    \end{align*}
    where $A^{\circ k}$ denotes $A$ to the Hadamard power of $k$.
    
    Since $M_S(\mu) \preceq I$, for every $i, j \in V_S$ such that $i\neq j$, we have 
    \begin{align*}
        2M_S(\mu)_{i,j} = (\*1_i + \*1_j)^\intercal M_S(\mu)(\*1_i + \*1_j) \leq (\*1_i + \*1_j)^\intercal I(\*1_i + \*1_j) = 2.
    \end{align*}
    Hence, $M_S(\mu)_{i,j} \leq 1$, which implies $\abs{M_S(\mu)_{i,j}} \leq 1$.
    This means for $\alpha \in [0,1]$, we can use the generalized binomial theorem, where the condition $\abs{M_S(\mu)_{i,j}} \leq 1$ guarantees the convergence of the series, and get
    \begin{align*}
      M_S(\nu) - I + J
      &= J + \sum_{k=1}^\infty \binom{\alpha}{k} (M_S(\mu) - I)^{\circ k}
       = J + \sum_{k=1}^\infty \binom{\alpha}{k} (-1)^{k} (I - M_S(\mu))^{\circ k}.
    \end{align*}
    According to the Schur product theorem, we have for every $k \geq 0$, 
    \begin{align*}
      (I - M_S(\mu))^{\circ k} \succeq 0.
    \end{align*}
    By the definition of the generalized binomial coefficients, we have
    \begin{align*}
      (-1)^{k}\binom{\alpha}{k} = (-1)^k \frac{\prod_{i=0}^{k-1}(\alpha - i)}{k!} \leq 0, \quad \forall k \in \mathbb{N}_+
    \end{align*}
    which implies that
    \begin{align*}
     M_S(\nu) - I &\preceq 0. \qedhere
    \end{align*}
\end{proof}

We also present a lemma that, similar as \cref{lem:log-concave-g_alpha}, establishes the preservation of strong log-concavity when the density of each subset is reweighted by a log-concave function on its size.
An immediate consequence is the rapid mixing of Glauber dynamics for the determinantal point process conditioned on subsets of size \emph{at most} $k$.

\begin{lemma}\label{lem:log-concave-varphi}
    Let $\XX \subseteq \{0,1\}^n$ be a non-empty downward closed set family, and let $\mu$ be a distribution fully supported on $\XX$.
    Suppose $\varphi: \mathbb{N}\mapsto \mathbb{R}_{\geq 0}$ is a log-concave function without internal zeros such that $\varphi(0) > 0$.
    Let $\nu$ be a distribution on $\XX$ defined as
    \begin{align*}
        \nu(S) \propto \varphi(|S|)\mu(S), \quad \forall S\in \XX.
    \end{align*}
    If the generating polynomial of $\mu$ is strongly log-concave, then the generating polynomial of $\nu$ is also strongly log-concave.
\end{lemma}
\begin{proof}
    By~\Cref{prop:log-concave}, it is sufficient to show that for any non-maximal $S\in \XX$,
    $$
        M_S(\mu) \preceq I \implies M_S(\nu) \preceq I .
    $$
    Fix a non-maximal $S\in \XX$ of size $k = |S|$. We have that
    \begin{align*}
        M_S(\nu) &= \frac{\varphi(k+2)\varphi(k)}{\varphi(k+1)^2}(M_S(\mu) + J - I) - (J - I)
        \\&= \frac{\varphi(k+2)\varphi(k)}{\varphi(k+1)^2} (M_S(\mu) - I) - \left( 1 - \frac{\varphi(k+2)\varphi(k)}{\varphi(k+1)^2} \right)J + I
        \\&\preceq - \left( 1 - \frac{\varphi(k+2)\varphi(k)}{\varphi(k+1)^2} \right)J + I 
        \preceq I,
    \end{align*}
    where the first inequality follows from the assumption $M_S(\mu) \preceq I$, and the second inequality follows from $\varphi(k+2)\varphi(k) \le \varphi(k+1)^2$, the log-concavity of $\varphi$.
\end{proof}


\bibliographystyle{alpha}
\bibliography{hc_random.bib}

\appendix
\section{Missing proofs in \texorpdfstring{\cref{subsec:stability}}{Section 3}}
\label{appendix-stab-proof}

\subsection{Equivalent definitions of spectral stability: Proof of~\texorpdfstring{\Cref{prop:equiv-SI-stability}}{Proposition 3.3}}

The proof of~\Cref{prop:equiv-SI-stability} relies on the following lemma.

\begin{lemma}\label{lem:property-f}
Let $\mu$ be a distribution over $\XX \subseteq \{0,1\}^n$. Furthermore, define the function $F:\+A \to \mathbb{R}$ as follows, where $\+A = \{\*m(\*\lambda * \mu) \mid \*\lambda \in \mathbb{R}^n_{>0} \} \subseteq [0,1]^{n}$:
\begin{align}\label{eq:def-fc}
    F_C(\*q) = \sum_{i \in [n]} q_i \log \frac{q_i}{p_i} - (q_i - p_i) - C \cdot \inf_{\nu \mid \*m(\nu) = \*q} D_{\mathrm{KL}}(\nu \parallel \mu),
\end{align}
where $\*p = \*m(\mu)$ denotes the mean vector of $\mu$. The following conditions are equivalent:
\begin{enumerate}
\item The function $F_C(\*q)$ is concave at $\*q = \*m(\*\lambda * \mu)$;
\item $\mathrm{Cov}(\*\lambda * \mu) \preceq C \cdot \diag\tp{\*m(\*\lambda * \mu)}$.
\end{enumerate}
Furthermore, $F_C(\*p) = 0$ and $\nabla F_C(\*p) = \*0_n$.
\end{lemma}
\begin{proof}
Let $\phi(\*q) = \inf_{\nu \mid \*m(\nu) = \*q} D_{\mathrm{KL}}(\nu \parallel \mu)$. By~\cite[Lemma 1]{bubeck2019entropic}, the Legendre dual of $\phi$ is $\psi(\*\lambda) = \log \tp{\sum_{S \in \XX} \exp\tp{\sum_{i \in S} \lambda_i} \mu(S)}$, implying that
\begin{align*}
    \nabla \phi(\*m(\*\lambda * \mu)) = \log (\*\lambda) \quad \text{and} \quad \nabla^2 \phi(\*m(\*\lambda * \mu)) = \tp{\mathrm{Cov}(\*\lambda * \mu)}^{-1}.
\end{align*}
Therefore, the Hessian matrix of $F(\*q)$ evaluated at $\*q = \*m(\*\lambda * \mu)$ is 
$$\nabla^2 F_C(\*m(\*\lambda * \mu)) = \tp{\diag(\*m(\*\lambda * \mu))}^{-1} - C\cdot \tp{\mathrm{Cov}(\*\lambda * \mu)}^{-1}.$$
Therefore, $F(\*q)$ is concave at $\*q = \*m(\*\lambda * \mu)$ if and only if $\mathrm{Cov}(\*\lambda * \mu) \preceq C \cdot \diag\tp{\*m(\*\lambda * \mu)}$.
The verification of $F_C(\*p) = 0$ and $\nabla F_C(\*p) = \*0_n$ is straightforward.
\end{proof}

We present the following useful corollary of~\Cref{lem:property-f}.

\begin{corollary}\label{cor:property-f}
   Let $\mu$ be a distribution over $\XX \subseteq \{0,1\}^n$. The following conditions are equivalent:
\begin{enumerate}
 \item $F_C(\*q)$ (defined in~\eqref{eq:def-fc}) is concave at $\*q = \*p$;
 \item For any distribution $\nu \ll \mu$, it holds \begin{align}\label{eq:desired-stability}
 \sum_{i\in [n]} p_i \tp{\frac{q_i}{p_i} - 1}^2 \le C \cdot D_{\chi^2}(\nu\parallel \mu),
 \end{align}
 where $\*p = \*m(\mu)$ and $\*q = \*m(\nu)$.
 \end{enumerate}
\end{corollary}
\begin{proof}
We first prove that the concavity of $F_C(\*q)$ at $\*q = \*p$ implies~\eqref{eq:desired-stability}. For any fixed $f: \XX \to \mathbb{R}$ with $\mu(f) = 0$, we let $f_\epsilon = 1 + \epsilon f$. The concavity of $F_C(\*q)$ implies
\begin{align}\label{eq:entropy-ineq-f}
    \sum_{i \in [n]} p_i \cdot \tp{\mu_{[n] \setminus i}(f_\epsilon )(1) \cdot \log \mu_{[n] \setminus i}(f_\epsilon)(1) - \mu_{[n] \setminus i}(f_\epsilon )(1) + 1} \le C \cdot \Ent[\mu]{f_{\epsilon}} + o(\epsilon^2).
\end{align}
This follows from that for any distribution $\nu$ absolutely continuous to $\mu$ and $h = \frac{\nu}{\mu}$, we have 
\begin{align*}
    \mu_{[n] \setminus i}(h)(1) = \frac{q_i}{p_i}\quad \text{and} \quad \Ent[\mu]{h} = D_{\mathrm{KL}}(\nu \parallel \mu),
\end{align*}
where $\*p=\*m(\mu)$ and $\*q=\*m(\nu)$ denote the mean vectors. By~\eqref{eq:entropy-ineq-f}, linearity of expectation, and the Taylor's expansion $(1+\epsilon) \log (1+ \epsilon) = \epsilon + \frac{1}{2}\epsilon^2 + o(\epsilon^2)$, it holds that
\begin{align*}
    \frac{1}{2} \cdot \tp{\sum_{i \in [n]} p_i \cdot \tp{\mu_{[n] \setminus i}(f)(1)}^2} \epsilon^2 \le \frac{C}{2} \cdot \Var[\mu]{f} \cdot \epsilon^2 + o(\epsilon^2).
\end{align*}
Combining the assumption $\mu(f) = 0$, we have
\begin{align}\label{eq:result}
    \sum_{i \in [n]} p_i \cdot \tp{\mu_{[n] \setminus i}(f)(1)- \mu(f)}^2 \le C\cdot \Var[\mu]{f}.
\end{align}
Finally, note that~\eqref{eq:result} holds for all $f^\star = f + t$, where $t \in \mathbb{R}$ is an arbitrary value. Thus, we may remove the constraint $\mu(f)=0$. Specifically, when $f=\frac{\nu}{\mu}$, it implies
\begin{align*}
    \sum_{i \in [n]} p_i \tp{\frac{q_i}{p_i}-1}^2 \le C\cdot D_{\chi^2}(\nu\parallel \mu).
\end{align*}
This concludes the proof. The converse follows similarly.
\end{proof}

\begin{remark}
    By~\Cref{prop:entropic-sufficient-condition}, the entropic stability with respect to the field dynamics is equivalent to $F_{C(\eta)}(\*q) \le 0$ for all conditional distributions $\mu^S$ and $\eta \in (0,1)$. This implies the concavity of $F_{C(\eta)}(\*q)$ at $\*q = \*p$ as $F_{C(\eta)}(\*p) = 0$. Together with~\Cref{cor:property-f}, this shows that entropic stability is a stronger condition than spectral stability.
\end{remark}
We are now ready to prove~\Cref{prop:equiv-SI-stability}.

\begin{proof}[Proof of~\Cref{prop:equiv-SI-stability}]
    The equivalence between~\Cref{item:spectral-1,item:spectral-2} has been implicitly established in~\cite{chen2024rapid}. 
    We first note that establishing the spectral stability with rate $C$ with respect to the field dynamics is equivalent to establishing~\eqref{eq:spec-stab-def} for all $f : \XX \to \mathbb{R}_{\ge 0}$ with $\E[(1-\eta)* \mu^S]{f} = 1$. This follows from that $\Var[\mu]{af + b} = a^2 \Var[\mu]{f}$ for all $a,b \in \mathbb{R}$.
    Therefore, we may assume $f = \frac{\dif \nu}{\dif (1-\eta) * \mu^S}$, the Radon-Nikodym derivative between an arbitrary $\nu$ and the stationary distribution $\mu$. By~\cite[Appendix A.1]{chen2024rapid}, the derivative of $\Var[Q_{\eta \to \theta}(S,\cdot)]{Q_{\theta \to 1}f}$ satisfies
    \begin{align*}
        \frac{\dif \Var[Q_{\eta \to \theta}(S,\cdot)]{f_\theta}}{\dif \theta} \Bigg|_{\theta = \eta^+} &= \frac{1}{1-\eta} \cdot \sum_{i \in [n] \setminus S} \Pr[R \sim Q_{\eta \to 1}(S,\cdot)]{i \in R} \cdot (f_\eta(S \cup \{i\}) - f_\eta(S))^2\\
        &= \frac{1}{1-\eta} \sum_{i \in [n] \setminus S} q_i \tp{\frac{q_i}{p_i} - 1}^2,
    \end{align*}
    where $\frac{\dif}{\dif \theta}|_{\theta = \eta^+}$ denotes the limit operator $\lim_{\theta \to \eta^+} \frac{1}{\theta - \eta}$, $\*p = \*m((1-\eta)* \mu^S)$ and $\*q = \*m(\nu)$. Combining $\Var[Q_{\eta \to 1}(S,\cdot)]{f} = D_{\chi^2}(\nu \parallel (1-\eta)*\mu^S)$, this proves the desired equivalence.

    The equivalence between~\Cref{item:spectral-2,item:spectral-3} follows from~\Cref{lem:property-f} and~\Cref{cor:property-f}.
    \end{proof}

We also include an alternative approach to establish the equivalence of~\Cref{item:spectral-1,item:spectral-3} in~\Cref{prop:equiv-SI-stability} using the localization schemes in~\cite{chen2022localization}. A similar proof strategy was used in proving the spectral stability for the coordinate-by-coordinate localization scheme.

Let $(Y_t)_{t \in [0,1]}$ be the continue time up walk (see \Cref{def:continuous-time-down-up-walk}).
For each $t \in [0, 1]$, we use $\nu_t := \-{Law}(Y_1 \mid Y_t)$.
The main idea in~\cite{chen2022localization} is that we can find a $\=R^n$ valued random variable $Z$ adapted to $Y_t$ such that $\nu_{t+h}$ can be well approximated by $\nu_t$ and $Z$.
The random variable $Z$ usually has good properties (e.g., independent coordinates) and is easier to analyze.
Let $Z$ be a $\=R^n$ valued random variable, we define a random distribution $\hat{\nu}_{t+h}$ as follow
\begin{align*}
  \hat{\nu}_{t+h}(x) := \nu_t(x)(1 + \inner{\*1_x - \*m(\nu_t)}{Z}),
\end{align*}
where we use the notation $\*1_x$ to denote the indicator function of set $x$ in $\set{0,1}^n$.
\begin{lemma}[\cite{chen2022localization}] \label{lem:CE-cauchy}
  Suppose $\E{Z \mid Y_t} = o(h)$ and for every function $f:2^n\to \=R$ we have
  \begin{align} \label{eq:CE-good-approx}
    \E{\E[\nu_{t+h}]{f}^2 \mid Y_t} = \E{\E[\hat{\nu}_{t+h}]{f}^2 \mid Y_t} + o(h),
  \end{align}
  then let $C_t := \-{Cov}(Z \mid Y_t)$, we have
  \begin{align*}
    \Var{\E{f(Y_1) \mid Y_{t+h}} \mid Y_t} \leq \norm{C_t^{1/2}\-{Cov}(\*1_{Y_1}\mid Y_t) C_t^{1/2}}_2 \Var{f(Y_1) \mid Y_t} + o(h).
  \end{align*}
\end{lemma}

\begin{remark}
  In~\cite{chen2022localization}, the condition \eqref{eq:CE-good-approx} is usually written as
  \begin{align*}
    \nu_{t+h}(x) = \nu_t(x)(1 + \inner{\*1_x - \*m(\nu_t)}{Z}) + o(h).
  \end{align*}
  This form is easy to be misunderstood as an event that happens alomost surely, which is not the case in our application.
  Hence we choose to use \eqref{eq:CE-good-approx} instead of this form.
\end{remark}

\begin{proof}[Alternative proof of~\Cref{prop:equiv-SI-stability}]
In order to apply \Cref{lem:CE-cauchy} to the field dynamics, we only need to find a good $Z$.
According to the definition of the continuous-time down-up walk, we have
\begin{align*}
  \Pr{Y_{t+h} = x \mid Y_t = y} &= o(h) + \begin{cases}
    0, & \text{if } \abs{x \setminus y} \geq 2; \\
    \frac{h}{1-t} \*m((1-t)*\mu^y)_i, & \text{if } x\setminus y = \set{i}; \\
    1 - \frac{h}{1-t}\sum_{i\in [n]\setminus y} \*m((1-t)*\mu^y)_i, & \text{if } x = y.
  \end{cases}
\end{align*}
Moreover, we have
\begin{align*}
        \nu_{t+h} = (1 - t - h) * \mu^{Y_{t+h}} = \tp{1 - \frac{h}{1-t}} * ((1-t) * \mu^{Y_{t+h}}) = \tp{1 - \frac{h}{1-t}} * \nu_t^{Y_{t+h} \setminus Y_t}.
\end{align*}

For a fixed $h > 0$, this means we can set $Z$ to be a random vector with independent coordinate such that $Z(j) = 0$ for $j \in Y_t$ and for $i \in [n]\setminus Y_t$,
\begin{align*}
  Z (i):= \begin{cases}
    \frac{1}{\*m(\nu_t)_i}, &\text{with prob. } \frac{h}{1-t}\;\*m(\nu_t)_i;\\
    \frac{-h}{1-t}, &\text{with prob. } 1 - \frac{h}{1-t}\;\*m(\nu_t)_i.
  \end{cases}
\end{align*}
It is known that $Z$ satisfies the requirements in \Cref{lem:CE-cauchy} (see~\cite{chen2022localization}).
According to this definition, we have
\begin{align*}
  \-{Cov}(Z \mid Y_t) &= \frac{h}{1-t}\; \diag\set{\E{\*1_{Y_1} \mid Y_t}- \*1_{Y_t}}^{-1} + o(h).
\end{align*}
For convenience, let $\Pi(Y_t) := \-{diag}\set{\E{\*1_{Y_1} \mid Y_t}- \*1_{Y_t}}$.
According to \Cref{lem:CE-cauchy}, this implies that
\begin{align*}
  \Var{\E{f(Y_1) \mid Y_{t+h}} \mid Y_t} &\leq \frac{h}{1-t} \norm{\Pi(Y_t)^{-1} \-{Cov}(\*1_{Y_1}\mid Y_t)}_2 \Var{f(Y_1) \mid Y_t} + o(h). 
\end{align*}
We finish the proof by comparing above equation with \eqref{eq:spec-stab-def}.
\end{proof}

\begin{proof}[Proof of~\Cref{lem:CE-cauchy}]
Let $g(Y_1) := f(Y_1) - \E{f(Y_1) \mid Y_t}$, so that $\E{g(Y_1) \mid Y_t} = \E{\nu_t}{g} = 0$.
This means our goal becomes
\begin{align*}
  \E{ \E{g(Y_1) \mid Y_{t+h}}^2 \mid Y_t} \leq \norm{C_t^{1/2}\-{Cov}(\*1_{Y_1}\mid Y_t) C_t^{1/2}}_2 \E{g(Y_1)^2 \mid Y_t} + o(h).
\end{align*}
This means we have
\begin{align*}
  \E{\E{g(Y_1) \mid Y_{t+h}}^2 \mid Y_t}
  &= \E{\E[\nu_{t+h}]{g}^2 \mid Y_t}
   \overset{\eqref{eq:CE-good-approx}}{=} \E{\E[\hat{\nu}_{t+h}]{g}^2 \mid Y_t} + o(h) \\
  (\text{def. of } \hat{\nu}_{t+h}) \quad
  &= \E{ \tp{\int_\Omega \inner{\*1_y - \*m(\nu_t)}{Z} g(y) \-d\nu_{t}(y) }^2 \mid Y_t} + o(h)\\
  &= \E{ \tp{\inner{\int_\Omega (\*1_y - \*m(\nu_t))g(y) \-d\nu_{t}(y) }{Z}}^2 \mid Y_t} + o(h).
\end{align*}
Let $v_t := \int_\Omega (\*1_y - \*m(\nu_t))g(y) \-d\nu_{t}(y)$, then we have
\begin{align*}
  \E{\E{g(Y_1) \mid Y_{t+h}}^2 \mid Y_t} + o(h)
  &= \E{\inner{v_t}{Z}^2 \mid Y_t} = \E{v_t^\intercal Z Z^\intercal v_t \mid Y_t} \\
  &= v_t^\intercal \E{Z Z^\intercal\mid Y_t} v_t = v_t^\intercal \-{Cov}(Z \mid Y_t) v_t
  = \norm{C_t^{1/2} v_t}_2^2.
\end{align*}
This implies
\begin{align*}
  \E{\E{g(Y_1) \mid Y_{t+h}}^2 \mid Y_t} + o(h)
  &= \norm{\int_\Omega C_t^{1/2} (\*1_y - \*m(\nu_t))g(y) \-d\nu_{t}(y)}_2^2 \\
  &= \sup_{\theta: \norm{\theta}_2 = 1} \inner{\int_\Omega C_t^{1/2} (\*1_y - \*m(\nu_t))g(y) \-d\nu_{t}(y)}{\theta}^2 \\
  &= \sup_{\theta: \norm{\theta}_2 = 1} \tp{\int_\Omega \inner{C_t^{1/2} (\*1_y - \*m(\nu_t))}{\theta}g(y) \-d\nu_{t}(y)}^2 \\
  (\text{by Cauchy–Schwarz ineq.}) \quad
  &\leq \sup_{\theta: \norm{\theta}_2 = 1} \tp{\int_\Omega \inner{C_t^{1/2} (\*1_y - \*m(\nu_t))}{\theta}^2 \-d \nu_t(y)} \tp{\int_\Omega g(y)^2 \-d\nu_{t}(y)} \\
  &= \norm{C_t^{1/2} \-{Cov}(\*1_{Y_1} \mid Y_t) C_t^{1/2}}_2 \E{g(Y_1)^2 \mid Y_t}. \qedhere
\end{align*}

\end{proof}

\subsection{Conditions of entropic stability: Proof of~\texorpdfstring{\Cref{prop:entropic-sufficient-condition}}{Proposition 3.7}}

\begin{proof}[Proof of~\Cref{prop:entropic-sufficient-condition}]
The equivalence between~\Cref{item:entropic-1,item:entropic-2} can be established in a similar manner to the proof of~\Cref{prop:equiv-SI-stability}.
Specifically, it suffices to compute the derivative of $\Ent[Q_{\eta \to \theta}(S,\cdot)]{Q_{\theta \to 1} f}$. With a similar argument, we may assume $f = \frac{\dif \nu}{\dif (1-\eta) * \mu^S}$ for some distribution $\nu \ll \mu$. By~\cite[Appendix A.1]{chen2024rapid}, the transition rule of $Q_{\eta \to \eta + h}(\cdot,\cdot)$ satisfies
\begin{align}\label{eq:field-kernel}
    Q_{\eta \to \eta + h}(S,T) =
    \begin{cases}        
    1-\frac{1}{1-\eta}\sum_{v \in [n] \setminus S}\Pr[R \sim Q(S,\cdot)]{v \in R} \cdot h + o(h), & \text{if }S = T;\\
        \frac{1}{1-\eta} \Pr[R \sim Q_{\eta \to 1}(S,\cdot)]{v \in R} \cdot h + o(h), & \text{if }\abs{T\setminus S}=1;\\
        o(h), & \text{otherwise.}
    \end{cases}
\end{align}
Similarly, the derivative of $f_\theta(S) = \E[Q_{\theta \to 1}(S,\cdot)]{f}$ satisfies
\begin{align*}
\frac{\dif f_\theta(S)}{\dif \theta} = \frac{\dif}{\dif \theta} \tp{\sum_{S \subseteq T} \frac{\mu(T)(1-\theta)^{\abs{T}-\abs{S}}}{\sum_{S \subseteq R} \mu(R) (1-\theta)^{\abs{R}-\abs{S}}} f(T)} = -\frac{1}{1-\eta} \tp{\E[T \sim Q_{\eta \to 1}(S,\cdot)]{\abs{T}\tp{f_{\eta}(T)-1}}}.
\end{align*}
Therefore, the entropy $\Ent[Q_{\eta \to \eta +h}(S,\cdot)]{f_{\eta + h}}$ satisfies
\begin{align*}
    \Ent[Q_{\eta \to \eta+h}(S,\cdot)]{f_{\eta + h}} &= \sum_{T \mid S \subseteq T} Q_{\eta \to \eta+h}(S,T) \cdot f_{\eta + h}(T) \log \frac{f_{\eta + h}(T)}{f_\eta(S)}\\
    &= \frac{1}{1-\eta} \sum_{v \in [n] \setminus S} \Pr[R \sim Q_{\eta \to 1}(S,\cdot)]{v \in R} \cdot f_\eta(S \cup \{v\}) \cdot \log \frac{f_\eta(S \cup \{v\})}{f_\eta(S)} \cdot h\\
    &- \frac{1}{1-\eta} \tp{\E[T \sim Q_{\eta \to 1}(S,\cdot)]{\abs{T}\tp{f_{\eta}(T)-1}}} \cdot h  +o(h)\\
    &= \frac{1}{1-\eta} \tp{\sum_{i \in [n] \setminus S}q_i \log \frac{q_i}{p_i} - (q_i - p_i)} \cdot h + o(h),
\end{align*}
where $\*p = \*m((1-\eta)*\mu^S)$ and $\*q = \*m(\*\nu)$ denotes the mean vector. Together with $\Ent[Q_{\eta \to 1}(S,\cdot)]{f} = D_{\mathrm{KL}}(\nu \parallel (1-\eta) * \mu^S)$, we prove the equivalence between~\Cref{item:entropic-1,item:entropic-2}.
To establish~\Cref{item:entropic-2} via~\Cref{item:entropic-3}, we use~\Cref{lem:property-f}. Fix any feasible pinning $S\subseteq [n]$, let $F_C(\*q)$ be the function for distribution $\mu^S$ defined in~\Cref{lem:property-f}. By~\Cref{item:entropic-3} and~\Cref{lem:property-f}, $F_C(\*q)$ is concave at its convex domain $\+A$, implying that $F_C(\*q) \le F_C(\*p) = 0$ as $\nabla F_C(\*p) = \*0_n$.
\end{proof}


\subsection{Equivalent definitions of log-concavity: Proof of \texorpdfstring{\cref{prop:log-concave}}{Proposition 3.8}}
\begin{proof}[Proof of \cref{prop:log-concave}]
We will establish the equivalence using the trickle-down theorem, \Cref{thm:cor-bound}, and characterizations of spectral stability and entropic stability.

\medskip
    \noindent``1 $\Rightarrow$ 2'': This follows from~\Cref{thm:cor-bound} and~\Cref{prop:equiv-SI-stability}.

\medskip
    \noindent``2 $\Leftrightarrow$ 3'': This follows from~\Cref{prop:equiv-SI-stability},~\Cref{prop:entropic-sufficient-condition}, and~\Cref{cor:property-f}. 
    
\medskip    
    \noindent``2 $\Rightarrow$ 4'': Note that a distribution $\mu$ over $\XX \subseteq \{0,1\}^n$ is log-concave if and only if
    \begin{align}\label{eq:condition-concavity}
        \forall z \in \mathbb{R}_{>0}^n,\quad \nabla^2 \log g_\mu(z) = \frac{g_\mu(z) \cdot \nabla^2 g_\mu(z) - (\nabla g_\mu(z))^{\otimes 2}}{g_\mu^2(z)} \preceq 0.
    \end{align}
    Moreover, the quantities $\tp{\frac{\nabla g_\mu(z)}{g_\mu(z)}}_i$ and $\tp{\frac{\nabla^2 g_\mu(z)}{g_\mu(z)}}_{i,j}$ satisfy
    \begin{align*}
    \tp{\frac{\nabla g_\mu(z)}{g_\mu(z)}}_i = \frac{1}{z_i} \Pr[S \sim z*\mu]{i \in S} \quad \text{and} \quad \tp{\frac{\nabla^2 g_\mu(z)}{g_\mu(z)}}_{i,j} =\begin{cases}
    \frac{1}{z_i z_j} \Pr[S \sim z*\mu]{i,j\in S}, & i \neq j;\\
    0, & \text{otherwise.}
    \end{cases}
    \end{align*}
    Therefore,~\eqref{eq:condition-concavity} is equivalent to $\mathrm{diag}\tp{z^{-1}} \tp{\mathrm{Cov}(z*\mu) - \mathrm{diag}(\*m(z * \mu))} \mathrm{diag}\tp{z^{-1}} \preceq 0$, where $z^{-1}$ denotes the vector obtained by taking inverse on each entry of $z$. This proves the implication.

\medskip
    \noindent``4 $\Rightarrow$ 1'': By previous calculations and~\eqref{eq:small-lambda}, we have
    \begin{align*}
        \nabla^2 \log g_{\mu^S}(\lambda \*1) = \lambda^{-2} \tp{\mathrm{Cov}(\lambda * \mu) - \mathrm{diag}\tp{\*m(\lambda * \mu)}} = \diag(\*r_S) (M_S - I) \mathrm{diag}(\*r_S) + o(1),
    \end{align*}
    where $\*r_S$ is defined in~\Cref{def:marginal-ratios}. Therefore, by taking $\lambda \to 0^+$, we have $M_S - I \preceq 0$.
\end{proof}

\end{document}